%% file: main.tex
\documentclass[lettersize,journal]{IEEEtran}
\usepackage{amsmath,amsfonts}
\usepackage{algorithmic}
\usepackage{algorithm}
\usepackage{array}

\usepackage{subcaption}
\usepackage{textcomp}
\usepackage{stfloats}
\usepackage{url}

\usepackage{verbatim}
\usepackage{graphicx}
\usepackage{cite}
\usepackage{xcolor}

\usepackage{pifont}
\usepackage{booktabs}
\usepackage{tcolorbox}
\usepackage{amsthm}
\usepackage{hyperref}
\usepackage{subcaption}
\usepackage{fontawesome5}

\definecolor{yellowgreen}{HTML}{ddffbf}
\definecolor{orange}{HTML}{fdb59c}
\definecolor{fedblue}{HTML}{ddffbf }

\begin{document}

\title{FedShield-LLM: A Secure and Scalable Federated Fine-Tuned Large Language Model}
\author{Md Jueal Mia and M. Hadi Amini,~\IEEEmembership{Senior Member,~IEEE}%
\thanks{Md Jueal Mia and M. Hadi Amini are with the Knight Foundation School of Computing and Information Sciences, Florida International University (FIU), Miami, FL 33199 USA. They are also with the Security, Optimization, and Learning for InterDependent networks laboratory (solid lab) at FIU (e-mail: mmia001@fiu.edu; moamini@fiu.edu).}%
}

\markboth{}%
{Shell \MakeLowercase{\textit{et al.}}: A Sample Article Using IEEEtran.cls for IEEE Journals}


\maketitle

\begin{abstract}
Federated Learning (FL) offers a decentralized framework for training and fine-tuning Large Language Models (LLMs) by leveraging computational resources across organizations while keeping sensitive data on local devices. It addresses privacy and security concerns while navigating challenges associated with the substantial computational demands of LLMs, which can be prohibitive for small and medium-sized organizations. FL supports the development of task-specific LLMs for cross-silo applications through fine-tuning but remains vulnerable to inference-related risks that threaten sensitive information. Prior studies have utilized Differential Privacy (DP) in LLM fine-tuning, which, despite being effective at preserving privacy, can degrade model performance. To overcome these challenges, we propose FedShield-LLM which integrates pruning with Fully Homomorphic Encryption (FHE) applied to Low-Rank Adaptation (LoRA) parameters. This combination enables secure computation over encrypted model updates and reduces the attack surface by deactivating less important LoRA parameters. Furthermore, optimized federated algorithms for cross-silo environments enhance scalability and efficiency. Parameter-efficient fine-tuning techniques like LoRA substantially reduce computational and communication overhead, making FL feasible for resource-constrained clients. Extensive experiments using Llama-2 models (7B and 13B) on four diverse datasets demonstrate that FedShield-LLM achieves superior collaborative performance and system efficiency compared to existing methods, supporting practical deployment across multiple domains. \\{\color{blue}The code and data are  available at \href{https://github.com/solidlabnetwork/fedshield-llm}{https://github.com/solidlabnetwork/fedshield-llm}.}
\end{abstract}

\begin{IEEEkeywords}
Federated Learning, Large Language Models, Security, Cross-Silo Applications, Fully Homomorphic Encryption, Low-Rank Adaptation.
\end{IEEEkeywords}

\input{Introduction}
\input{Literature_Review}
\input{Preliminaries}
\input{Methodology}

\input{Experiments_and_Result_Analysis}

\input{Discussion}

\section{Conclusion}
In this study, we proposed a secure and efficient mechanism, FedShield-LLM, for fine-tuning LLMs in FL by integrating FHE with unstructured pruning. As part of FHE, the CKKS encryption scheme ensures that model parameters remain encrypted throughout training and aggregation, protecting client data privacy. All model parameters are encrypted layer-wise and shared with the server to keep them secure during the communication and aggregation process. Unstructured pruning enhances security by deactivating less significant weights in LoRA parameters, reducing the attack surface and mitigating risks from inference attacks. Experimental results demonstrate that the proposed framework outperforms existing methods in text generation performance, while maintaining robust privacy guarantees and computational efficiency. This makes the approach suitable for real-world applications in sensitive domains. Future work will focus on addressing other categories of adversarial attacks during the training phase to further enhance the robustness of the framework in distributed environments.

\section*{Acknowledgement}
This work is based upon the work supported by the National Center for Transportation Cybersecurity and Resiliency (TraCR) (a U.S. Department of Transportation National University Transportation Center) headquartered at Clemson University, Clemson, South Carolina, USA. Any opinions, findings, conclusions, and recommendations expressed in this material are those of the author(s) and do not necessarily reflect the views of TraCR, and the U.S. Government assumes no liability for the contents or use thereof.

\section*{Disclaimer} This paper presents research findings on privacy-preserving federated learning with large language models and is intended solely for academic and research purposes. The methods and results described should not be interpreted as commercial, legal, or professional advice regarding the deployment of AI systems. Readers are encouraged to exercise caution when applying these techniques, particularly in sensitive or high-stakes contexts, as experimental performance may not translate to real-world effectiveness or safety. The authors’ claims and conclusions reflect their own analysis of the data and do not necessarily represent the views of their institutions, sponsors, or the broader research community. The authors disclaim any responsibility for any downstream use or misuse of the proposed methods.

\bibliographystyle{IEEEtran}
\bibliography{references}

\begin{IEEEbiography}[{\includegraphics[width=1in,height=1.25in,clip,keepaspectratio]{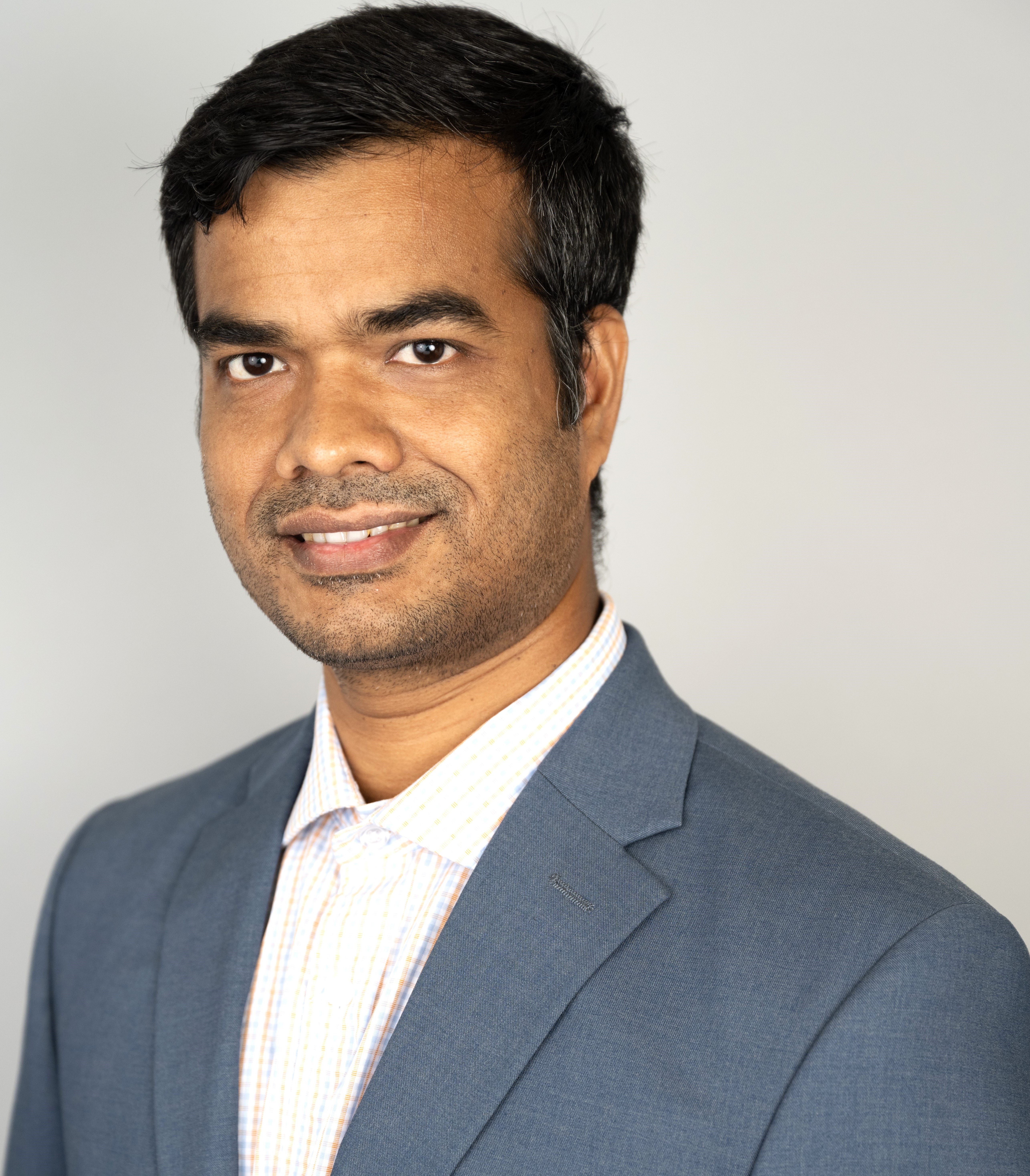}}]
{Md Jueal Mia}~ is pursuing his PhD in Computer Science at Security, Optimization, and Learning for InterDependent networks lab (solid lab), Knight Foundation School of Computing and Information Sciences, Florida International University (FIU) under supervision of Dr. M. Hadi Amini. 
Prior to that, he
received the B.S. and M.S. degrees in computer science and engineering from Jahangirnagar University, Bangladesh. He has published more than 30 peer-reviewed journal and conference publications. Before joining FIU, he served as a Faculty Member with the Department of Computer Science and Engineering, Daffodil International University, Dhaka, Bangladesh, for more than six years. His research interests include privacy and security issues in federated learning, distributed machine learning, and  large language models. 
\end{IEEEbiography}

\begin{IEEEbiography}[{\includegraphics[width=1in,height=1.25in,clip,keepaspectratio]{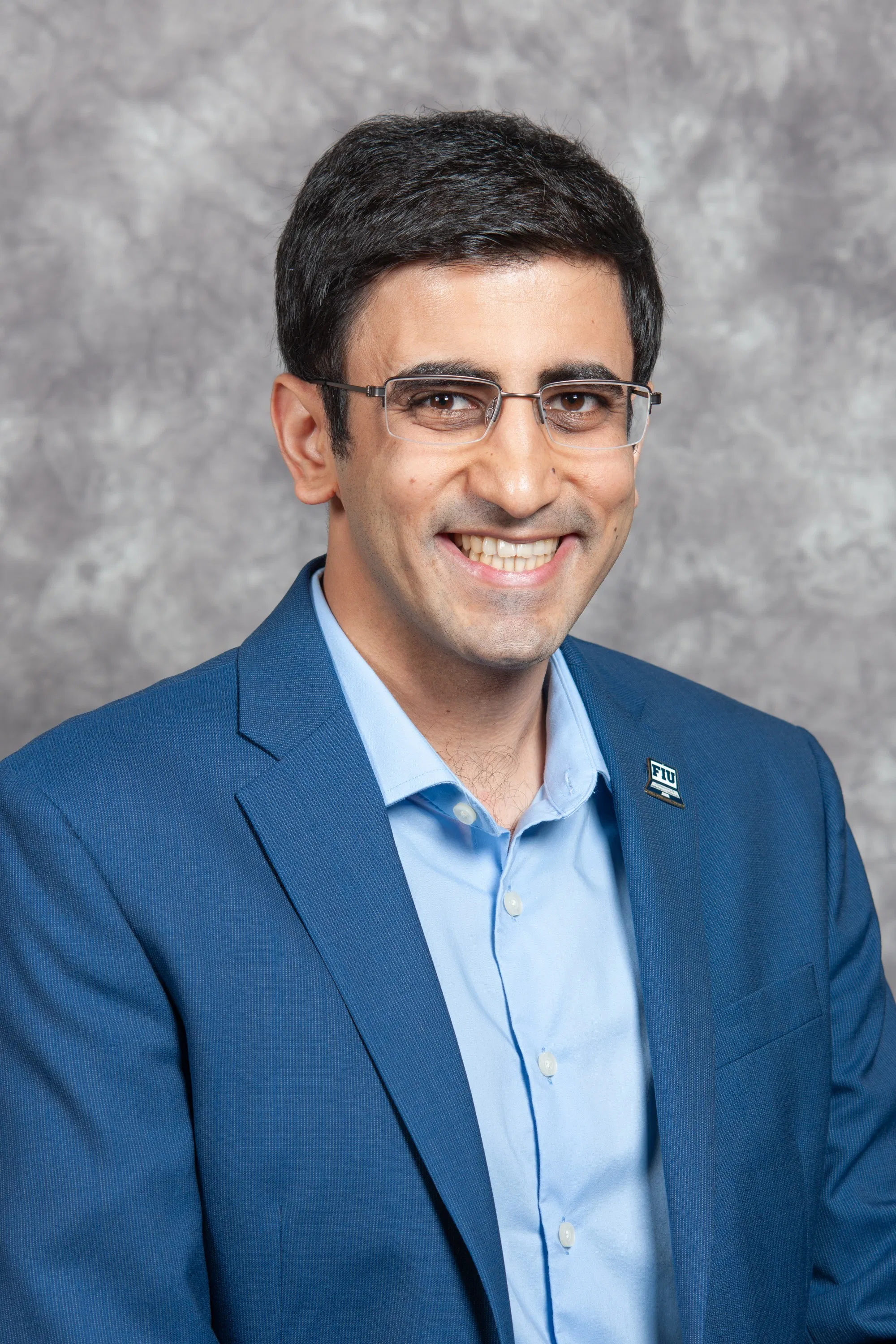}}]{M. Hadi Amini (Senior Member ACM, Senior Member IEEE) }      is an Associate Professor of Computer Science in the Knight Foundation School of Computing and Information Sciences (KFSCIS) and founding director of Security, Optimization, and Learning for InterDependent Networks Laboratory (www.solidlab.network) at Florida International University (FIU). He also serves as  Associate Director of the USDOT National University Transportation Center for Transportation Cyber Security and Resiliency. Prior to that, he received his Ph.D. in Electrical and Computer Engineering from Carnegie Mellon University in 2019. He conducts research in (decentralized) learning and optimization algorithms, their security and privacy vulnerabilities, and applications in cyber-physical systems security and  resilience. He is the recipient of  ``2025 IEEE Big Data Security Junior Research Award'' for contributions to Big Data Security in Cyber-Physical Systems, Best Paper Award from ``2019 IEEE Conference on Computational Science \& Computational Intelligence'', 2021 Best Journal Paper Award from ``Springer Nature Operations Research Forum Journal'', 2024 “FIU Top Scholar Award, Research and Creative Activities, Junior Faculty with Significant Grants (Sciences)”, 2025 ``FIU Knight Foundation School of Computing and Information Sciences Excellence in Applied Research Award'', and 2023 FIU “Faculty Senate Excellence in Teaching Award”, among others. He was also selected as one of the 2025 Rising Star (Sciences) by the Academy of Science, Engineering and Medicine of Florida (ASEMFL). He is a Senior Member of ACM and IEEE. He serves as an Associate Editor of \textit{IEEE Transactions on Information Forensics and Security} and \textit{IEEE Transactions on Machine Learning in Communications and Networking}. 
 (Homepage: www.hadiamini.com)
\end{IEEEbiography}

\end{document}

%% file: Introduction.tex
\section{Introduction}
\IEEEPARstart{L}{arge} Language Models (LLMs) such as GPT~\cite{brown2020language}, PaLM~\cite{chowdhery2023palm}, and ChatGPT~\cite{OpenAI2023} have revolutionized natural language processing by achieving state-of-the-art results across a wide array of language understanding and generation tasks. These foundation models are trained on large-scale public datasets, which enables strong performance in dialogue generation, summarization, and question answering. However, as LLMs are increasingly fine-tuned for specific domains such as healthcare and finance, they must incorporate sensitive or proprietary data, raising critical privacy and compliance concerns. Regulations like the GDPR and CCPA restrict data sharing and impose legal obligations on organizations handling private information~\cite{baik2020data, boyne2018data, das2024security}. Moreover, recent studies show that LLMs can memorize and leak training data~\cite{carlini2021extracting}, which presents significant risks when models are trained or fine-tuned on confidential inputs. This makes it imperative to develop solutions that preserve data locality and ensure privacy during domain adaptation.

Training and fine-tuning LLMs also pose substantial computational challenges. Models like GPT-3 require hundreds of GPU-years and millions of dollars in compute resources~\cite{LambdaLabs2020}, limiting accessibility to large tech firms. Even with the release of open-source LLMs like LLaMA~\cite{touvron2023llama}, resource bottlenecks remain, especially for smaller organizations or in distributed collaborative scenarios. Traditional full-model fine-tuning methods involve transmitting large model updates or datasets, which can be impractical in resource-constrained environments typically found in centralized training scenarios. FL \cite{mcmahan2017communication} addresses these challenges by enabling decentralized training, where clients perform local updates on private data and only share model gradients or parameters with a central server. FL can  enhance privacy and  efficiency of LLMs \cite{amini2025distributed}. FL has been explored in resource-constrained and edge scenarios~\cite{imteaj2021survey}, and recent works such as  OpenFedLLM~\cite{r3_ye2024openfedllm} and FederatedScope-LLM~\cite{r2_kuang2024federatedscope} have shown that FL can be adapted for LLM instruction tuning and domain-specific applications.

However, standard FL remains vulnerable to efficiency and security issues. Transmitting full-model updates for large models is bandwidth-intensive and often infeasible. Additionally, model updates can leak private information via inference or gradient inversion attacks (GIA)~\cite{zhu2019deep, gu2022cs, petrov2024dager}. To address the communication overhead, researchers have proposed parameter-efficient tuning methods such as LoRA~\cite{hu2022lora}, which inserts low-rank matrices into transformer layers to reduce trainable parameters. LoRA enables clients to only update lightweight adapter layers instead of the full model, drastically reducing computation and communication costs. This makes LoRA highly compatible with federated fine-tuning frameworks~\cite{wang2024flora}. Nonetheless, privacy concerns persist even with reduced update sizes, as model updates can still be exploited by malicious servers or compromised clients.

To ensure end-to-end privacy in federated LLM training, several advanced privacy-preserving techniques have been proposed. Differential Privacy (DP)~\cite{dwork2014algorithmic} introduces noise to model updates to provide statistical guarantees but often compromises model accuracy, especially in high-dimensional settings like LLMs~\cite{liu2023differentially, xu2024dp}. FHE~\cite{cheon2017homomorphic} enables computation over encrypted data, offering stronger confidentiality without requiring noise, though at the cost of increased computation. Secure Multi-Party Computation (MPC)~\cite{bonawitz2017practical} enables joint computation without revealing data but also faces scalability limitations. Meanwhile, model pruning~\cite{wu2021adversarial} can improve efficiency and mitigate attack surfaces by eliminating redundant parameters. Although most prior privacy-preserving LLM frameworks focus on inference-time security in centralized settings, they offer complementary insights relevant to secure model design. For instance, PrivacyAsst~\cite{zhang2024privacyasst} leverages HE to safeguard sensitive user inputs during tool-based LLM inference, illustrating how encrypted computation can be integrated into LLM workflows at runtime. Similarly, InferDPT~\cite{tong2025inferdpt} applies local DP to perturb user prompts submitted to black-box LLMs and reconstructs coherent outputs using a local extraction model.
Orthogonal to these privacy-preserving inference frameworks, decoding methods such as nucleus/top-p, Min-p, and geometry-aware Top-W improve generation behavior by shaping the next-token distribution at inference time~\cite{holtzman2020curious,nguyen2025turning,davoodi2026geometry}.

Building on these insights, we propose FedShield-LLM, a novel federated fine-tuning framework that combines FL, LoRA, FHE, and pruning to ensure scalable, secure, and regulation-compliant training of LLMs across decentralized private datasets. The primary motivation is to develop a framework that enables efficient LLM training on edge devices using distributed learning methods. Centralized training is computationally expensive and also raises serious data privacy concerns. In addition, real-time distributed training and inference introduce significant data security challenges. Existing DP based approaches suffer from the inherent trade-off between privacy and utility, where added noise degrades model performance. In this work, we detail the design of FedShield-LLM and demonstrate its effectiveness in achieving high performance with strong privacy guarantees. The main contributions of our study are listed below.

\begin{itemize}
    \item We propose \textbf{FedShield-LLM}, a novel federated fine-tuning framework for LLMs that jointly integrates FHE, LoRA, and unstructured pruning. To the best of our knowledge, this is the first work to explore the combination of FHE and pruning in federated LLM fine-tuning, enabling secure aggregation in the encrypted domain while mitigating privacy leakage during both training and inference phases. Our framework provides robustness against inference attacks under adversarial settings, including an honest-but-curious server.

    \item FedShield-LLM improves feasibility by leveraging LoRA to update only lightweight adapter layers instead of full model parameters, significantly reducing computational and memory overhead. This design makes secure encrypted fine-tuning practical even for large-scale LLMs and resource-constrained environments.

    \item We validate FedShield-LLM through extensive experiments using base models \textit{meta-llama/Llama-2-7b-hf} and \textit{meta-llama/Llama-2-13b-hf} across diverse datasets, including \textit{medalpaca/medical\_meadow\_medical\_flashcards}, \textit{vicgalle/alpaca-gpt4}, \textit{TIGER-Lab/MathInstruct}, and \textit{FinGPT/fingpt-sentiment-train}. The results demonstrate that our framework consistently achieves strong performance while providing enhanced security compared to existing approaches.
\end{itemize}

The remainder of this paper is structured as follows: Section II provides a comprehensive review of the related literature. Section III introduces the preliminaries essential for understanding the proposed method. Section IV details the proposed FedShield-LLM methodology, followed by Section V, which presents the experimental setup and results. Section VI offers further discussion, including limitations and implications. Finally, Section VII concludes the paper and outlines directions for future work.

%% file: Literature_Review.tex
\section{Literature Review}
\label{sec:litreview}

Recent advancements in machine learning have stimulated significant interest in combining FL with LLMs to enable collaborative model development while preserving data privacy. Researchers have explored several approaches to make this integration both practical and secure. Prior studies have examined federated and distributed fine-tuning strategies, and have also integrated differential privacy, homomorphic encryption (HE), or secure multi-party computation to reduce information leakage during distributed training.
While each of these methods contributes valuable capabilities, they also come with limitations highlighting the need for a comprehensive framework like FedShield-LLM that unifies efficiency, privacy, and robustness.

 Recent works have introduced efficient frameworks for distributed and federated LLM adaptation. OpenFedLLM~\cite{r3_ye2024openfedllm} proposes a comprehensive pipeline integrating federated instruction tuning and value alignment for adapting models like LLaMA2-7B across multiple clients without data centralization. It supports various FL strategies, including FedAvg, and demonstrates that federated fine-tuned models can match the performance of central LLMs like GPT-4 on specialized tasks. FedBiOT~\cite{r4_wu2024fedbiot} addresses client-side limitations using a bi-level optimization scheme in which clients fine-tune compressed LLM emulators (lightweight LoRA modules), while servers maintain alignment with full models. This drastically reduces local resource demands. FederatedScope-LLM~\cite{r2_kuang2024federatedscope} is a modular open-source framework that provides unified support for various LLM FL scenarios, including adapter-based tuning and diverse FL algorithms. Separately, FlexLoRA~\cite{bai2024federated} proposes an adaptive LoRA strategy in heterogeneous FL environments: each client trains with a custom LoRA rank, and the server merges these updates using singular value decomposition (SVD) to synthesize a full-rank global model. This avoids bottlenecks from underpowered clients and improves performance across diverse clients and tasks. FedCoLLM~\cite{fan2024fedcollm} further extends this line by co-tuning LLMs and smaller SLMs using LoRA, Secure Aggregation~\cite{bonawitz2016practical}, and Knowledge Distillation~\cite{hinton2015distilling}, achieving less than 0.25\% of full model communication while preserving performance. Yun et al.~\cite{yun2023privacy} introduce a hierarchical clustered sampling method to address non-IID data, combining within-cluster aggregation and multinomial participation to improve fairness and stability. A recent survey~\cite{wu2025survey} categorizes distributed fine-tuning strategies like knowledge distillation and split learning, which offer trade-offs between privacy, performance, and communication. While these frameworks enable distributed LLM adaptation, most assume honest clients and servers, and lack robust protections against inference attacks.

DP has been widely explored for protecting data contributions in FL. DP-LoRA~\cite{liu2023differentially} applies Gaussian noise to LoRA adapter weights, enabling formal $(\epsilon,\delta)$ privacy guarantees while maintaining model accuracy. This works well due to LoRA’s compressed structure, which reduces the noise scale required. DP-DyLoRA~\cite{xu2024dp} extends this by dynamically adjusting adapter ranks, integrating rank-sensitive noise mechanisms to further optimize the privacy-utility tradeoff. Cross-domain evaluations (e.g., speech, vision, text) show that DP-DyLoRA maintains  $<2\%$ loss in performance even at $\epsilon=2$ with one million clients. Yu et al.~\cite{yu2021differentially} found that parameter-efficient fine-tuning (PEFT) techniques (e.g., adapters) inherently offer stronger DP tradeoffs versus full fine-tuning. DP offers strong protection only when the privacy budget is extremely small, but such settings significantly degrade model utility. This fundamental privacy–utility trade-off limits DP’s ability to prevent memorization and gradient-based leakage, making it insufficient as a standalone privacy mechanism.

 HE offers an orthogonal privacy defense by allowing secure aggregation of encrypted model updates. Frameworks such as PrivTuner ~\cite{li2024privtuner} a centralized fine-tuning framework that integrates FHE with LoRA to enable secure and efficient fine-tuning of AI foundation models. Unlike FL, which relies on decentralized training, PrivTuner ensures data privacy by performing computations on encrypted client data directly on the server, eliminating the need for raw data transmission while maintaining model performance. FHE offers robust security against inference attacks during LLM fine-tuning in an FL environment. However, security challenges can arise in cross-silo FL scenarios, particularly when the server evaluates the performance of the LLM. In such cases, honest-but-curious clients or servers may attempt to infer sensitive information by analyzing the LoRA parameters, which contain knowledge learned from private data.

 MPC protocols such as secret sharing and garbled circuits enable joint computation without revealing private data. Google’s secure aggregation protocol~\cite{bonawitz2017practical} masks individual updates so only the final sum is revealed to the server. These systems strengthen client privacy but suffer from high latency and are hard to scale to models with billions of parameters.

Beyond LLM focused works, recent enterprise FL studies also explore privacy preserving learning. For example, FedAnil \cite{fotohi2024decentralized} integrates CKKS FHE, blockchain randomness, and clustering (cosine similarity and affinity propagation) to defend against poisoning, collusion, and inference attacks under non IID enterprise data. A related lightweight framework \cite{fotohi2024lightweight} similarly combines CKKS FHE with gradient clustering and entropy based compression to reduce communication cost while preserving privacy. Although these methods demonstrate the benefit of coupling HE with compression, they are developed for standard deep models and not large language models.

Unlike DP-LoRA, which uses noise to guarantee privacy at some cost to utility, our method avoids accuracy loss by using FHE for exact aggregation. Based on prior work and our experimental observations, DP-based LoRA methods introduce a clear privacy--utility trade-off, where noise injection into model updates can lead to performance degradation and training instability, especially under strong privacy constraints. Moreover, whereas PrivTuner applies FHE and LoRA in a single-server setting for data sharing, FedShield-LLM operates in a federated multi-client environment and introduces pruning to further reduce attack surface and overhead. In contrast to DP-based approaches, FedShield-LLM adopts a fundamentally different paradigm by leveraging FHE to perform secure aggregation directly in the encrypted domain without perturbing model updates, thereby preserving gradient integrity and model performance while ensuring strong cryptographic privacy guarantees. Table~\ref{tab:comparison_methods} provides a concise comparison of existing LLM fine-tuning approaches, highlighting differences in distribution settings, underlying security mechanisms, privacy guarantees, attack resiliency, and utility tradeoffs. It integrates LoRA-based efficient fine-tuning with FHE and model pruning: LoRA minimizes the size of model updates, making FHE computationally feasible, while pruning sparse the model and reduces the attack surface. As a result, FedShield-LLM achieves a more favorable privacy--utility trade-off compared to DP-based methods, clearly demonstrating its effectiveness and practical advantage. Together, these techniques form a comprehensive framework for secure, efficient, and robust collaborative fine-tuning of LLMs in adversarial environments.

\begin{table*}[t]
\centering
\caption{Comparison of Secure LLM Fine-tuning Methods}
\resizebox{\textwidth}{!}{
\begin{tabular}{l|c|c|c|c|c|c|c}
\hline
\textbf{Name} &
\textbf{Security Mechanism} &
\textbf{Distributed} &
\textbf{Data Privacy} &
\textbf{Inference Attack Resiliency} &
\textbf{Scalability} &
\textbf{Computational Overhead} &
\textbf{Performance Tradeoff} \\ \hline

OpenFedLLM \cite{r3_ye2024openfedllm}
& None
& \ding{51}
& \ding{51}
& \ding{55}
& \ding{51}
& No overhead (Vanilla baseline)
& \ding{51} High Utility \\ \hline

DP-LoRA \cite{liu2023differentially}
& DP
& \ding{51}
& \ding{51}
& \ding{51}
& \ding{51}
& Low overhead
& \ding{55} Accuracy drops when $\epsilon$ small \\ \hline

PrivTuner \cite{li2024privtuner}
& HE
& \ding{55}
& \ding{51}
& \ding{55}
& \ding{55}
& High overhead due to data encryption
& \ding{55} Heavy compute; Not Federated \\ \hline

FedShield-LLM (Ours)
& HE + Pruning
& \ding{51}
& \ding{51}
& \ding{51}
& \ding{51}
& Moderate overhead for encryption/decryption operations
& \ding{51} Scalable; negligible Utility Drop \\ \hline

\end{tabular}
}
\label{tab:comparison_methods}
\end{table*}

%% file: Preliminaries.tex
\section{Preliminaries} 

\subsection{Federated Fine-Tuning of LLMs}\label{sec:federated-llm}

FL enables a set of $N$ clients $\{C_1, C_2, \ldots, C_N\}$ to collaboratively fine-tune a shared global LLM with parameters $w_t$ at round $t$, using decentralized datasets $\{D_1, D_2, \ldots, D_N\}$, without sharing raw data. Each client $C_i$ receives the current global model $w_t$ and performs local training to compute a LoRA-based model update $\Delta w_i$, which represents the client-specific parameter change:
\begin{equation}
\Delta w_i := \arg\min_{\Delta w} \; L_i\big(f_{w_t} + \Delta w, D_i\big),
\end{equation}
where $L_i$ is the local loss computed over dataset $D_i$. After local updates, the server aggregates the received client updates using a method such as FedAvg~\cite{mcmahan2017communication}. Assuming equal weight for each client, the global model is updated as:
\begin{equation}
w_{t+1} = w_t + \frac{1}{N} \sum_{i=1}^{N} \Delta w_i,
\end{equation}
and the updated model $w_{t+1}$ is redistributed to clients for the next communication round.

In practice, transmitting full gradients or model updates is resource-intensive, especially for large LLMs. To mitigate this, recent FL methods adopt PEFT strategies such as LoRA~\cite{hu2022lora}, which restrict training to low-dimensional subspaces within transformer layers. In LoRA, the update $\Delta w_i$ for each weight matrix is represented as a low-rank factorization: $\Delta w_i = A_i \cdot B_i$, where $A_i \in \mathbb{R}^{d \times r}$ and $B_i \in \mathbb{R}^{r \times k}$. Since $r \ll \min(d, k)$, this reduces the number of trainable parameters

\subsection{Threat Model}\label{sec:threat-model}
Even though raw data never leaves local devices in FL, the exchanged model updates $\Delta w_i$ can leak sensitive information about client $C_i$’s dataset $D_i$. We consider an inference attack threat model wherein an adversary observes gradients or model updates and attempts to infer private data. The adversary can be an honest-but-curious server (a semi-honest aggregator that faithfully executes FL but analyzes received updates for information) or dishonest clients that deviate from the protocol. For example, a curious central server could try to invert gradients to reconstruct a client's training examples. Attacks such as Deep Leakage from Gradients (DLG)~\cite{zhu2019deep} and other GIA~\cite{geiping2020inverting} have demonstrated that given $\Delta w_i$, an attacker can often reconstruct the original inputs used to compute that update. Similarly, an adversary might perform membership inference to determine if a certain sample was in $D_i$. 

We can formalize the privacy breach by the probability $\Pr[\mathcal{A}(\Delta w_i) = x]$ that an adversary $\mathcal{A}$, given model update $\Delta w_i$, correctly infers a private data point $x \in D_i$. An effective attack means this probability is much higher than random chance (i.e., $\gg 1/|X|$ for input domain $X$):
\begin{equation}
\Pr\big[\mathcal{A}(\Delta w_i) = x\big] \gg \frac{1}{|X|}~,
\end{equation}
indicating a significant privacy risk.

To mitigate such threats, a strong privacy-preserving approach is the cryptographic encryption of model updates using FHE. FHE schemes—such as CKKS, based on the Ring Learning With Errors (Ring-LWE) problem~\cite{cheon2017homomorphic}; enable computations directly over encrypted data. In the context of FL, each client encrypts its model update $\Delta w_i$ and sends it to the server, which then performs aggregation (e.g., summing or averaging) without ever decrypting the individual updates. The result is a ciphertext of the aggregated model, which can then be decrypted by an authorized party. This process ensures that even an honest-but-curious server or external adversary cannot infer any client’s private information from the model updates. Since FHE prevents direct access to individual gradients or parameter updates, it is highly effective at mitigating gradient inversion attacks and membership inference threats, offering strong confidentiality guarantees in adversarial FL settings.

\subsection{Applications}\label{sec:applications}
Secure federated fine-tuning of LLMs is useful in domains where data is sensitive, regulated, or proprietary. Below we highlight key application areas.

\subsubsection{Healthcare}
Hospitals can collaboratively fine-tune LLMs on electronic health records (EHR) to build clinical QA systems or decision-support tools while complying with privacy regulations such as HIPAA and GDPR. FL enables models to benefit from multi-institution data without sharing raw patient records, and prior work demonstrates strong performance on clinical NLP tasks under federated settings~\cite{zhang2025federated}.

\subsubsection{Finance}
Banks and financial institutions can jointly train LLMs on sensitive data such as transaction logs, fraud alerts, or risk reports without exposing proprietary information. Federated fine-tuning enables richer financial modeling from distributed data silos while maintaining confidentiality, and FL has been shown to improve fraud detection and credit risk prediction with appropriate privacy safeguards~\cite{liu2023efficient}.

\subsubsection{Autonomous Vehicles (AV)}
Federated learning allows intelligent vehicles to train perception and decision-making models using learned representations from diverse driving environments without transmitting raw sensor data. This improves robustness while preserving user privacy, supporting tasks such as lane detection, obstacle avoidance, and traffic sign recognition~\cite{zhang2024federated}. Federated fine-tuning of LLMs may further enhance AV systems by

\begin{table*}[htbp]
\centering
\caption{Notations}
\label{tab:notations}
\begin{tabular}{llll}
\toprule
\textbf{Notation} & \textbf{Representation} &
\textbf{Notation} & \textbf{Representation} \\
\midrule

$N$ & Total number of clients &
$R$ & Total number of communication rounds \\

$t$ & Current communication round &
$D_i$ & Local dataset of client $i$ \\

$D$ & Set of all local datasets &
$p_i$ & Aggregation weight for client $i$ \\

$p_t$ & Pruning rate at round $t$ &
$p_0,p_{\text{target}}$ & Initial and target pruning rates \\

$t_{\text{eff}}, t_{\text{target}}$ & Pruning schedule params &
$n_t$ & Participating clients at round $t$ \\

$w_t$ & Global model at round $t$ &
$w_i$ & Local model of client $i$ \\

$\Delta w_i$ & LoRA update of client $i$ &
$\Delta w_i^r$ & LoRA update at round $r$ \\

$\Delta w_i^p$ & Pruned LoRA update &
$\Delta w_{i,\text{prune}}^r$ & Pruned update at round $r$ \\

$A_i,B_i$ & LoRA low-rank matrices &
$d,r$ & LoRA dimensions \\

$m_i$ & Pruning mask for client $i$ &
$e_i^r$ & Pruning error at round $r$ \\

$\epsilon_r$ &
  Bound on $\|e_i^r\|$ &
$\bar{w}_t$ & Aggregated global update \\

$\bar{\Delta}^r$ &
  Aggregated pruned update &
$\bar{\Delta}_{\text{full}}^r$ &
  Aggregated unpruned update \\

$\bar{e}^r$ &
  Aggregated pruning error &
$F(w)$ & Global FL objective \\

$F_i(w)$ & Local objective function &
$\eta$ & Learning rate \\

$L$ &
  Smoothness constant &
$\sigma^2$ &
  Gradient variance bound \\

$\min_{0\le r < R}\|\nabla F(w^r)\|^2$ &
  Stationarity measure &
$C$ &
  Constant depending on $L,\sigma^2$ \\

$\mu$ &
  PL/strong convexity const. &
\textsf{CKKS} & Homomorphic encryption scheme \\

\textsf{FHE} & Fully Homomorphic Encryption &
$\mathcal{C}$ & CKKS encryption context \\

$n_c$ & \# ciphertexts for LoRA packing &
$N_{\text{poly}}$ & CKKS polynomial degree \\

$\mathcal{A}$ & Adversarial algorithm &
$\Pr[\mathcal{A}(\Delta w_i)=x]$ &
  Probability of data reconstruction \\

\bottomrule
\end{tabular}
\end{table*}

%% file: Methodology.tex
\section{Methodology} 

\subsection{Method}
\label{sec:Methodology}

\begin{figure*} [!t]
    \centering
    \includegraphics[width=0.80 \linewidth]{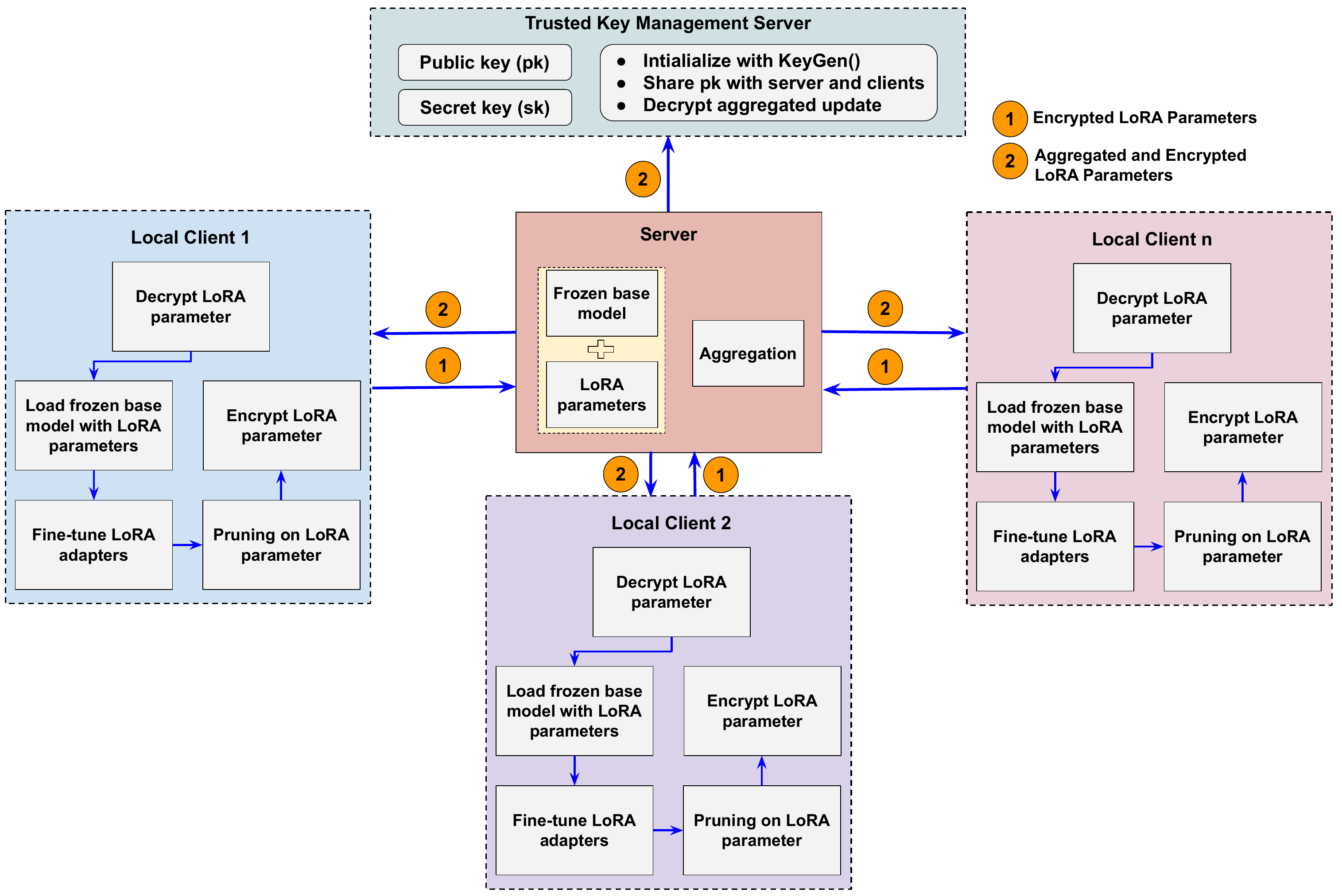}
    \caption{Overview of our proposed framework.}
    \label{fig:figure_1}
\end{figure*}

Fine-tuning large language models (LLMs) in federated environments introduces three key challenges: (i) preserving data privacy across distributed clients, (ii) ensuring computational efficiency on resource-constrained devices, and (iii) maintaining scalability for large models. To address these challenges, we propose FedShield-LLM, a unified framework that integrates FHE, LoRA, and unstructured pruning. Our mechanism ensures a secure and efficient fine-tuning process, as illustrated in Figure \ref{fig:figure_1} and detailed in Algorithm \ref{alg:fedshield_llm}. All notations are presented in Table \ref{tab:notations}.

At a high level, FHE enables secure aggregation of client updates without exposing raw parameters, LoRA reduces the number of trainable parameters to make fine-tuning efficient, and pruning limits the attack surface while improving communication efficiency. Our framework builds upon the federated LLM fine-tuning paradigm introduced in OpenFedLLM~\cite{r3_ye2024openfedllm}, and extends it with strong privacy and efficiency guarantees.

As illustrated in Fig.~\ref{fig:figure_1}, the server initializes a frozen base LLM and distributes only lightweight LoRA adapter parameters to selected clients. Each client performs local fine-tuning on private data, followed by sparsification and encryption of updates before transmission. The server aggregates encrypted updates without accessing any individual client information and broadcasts the aggregated result back to clients. This process iterates across multiple communication rounds, after which the final LoRA adapters are merged with the base model to obtain the task-specific LLM. Finally, the third-party server can decrypt and merge the LoRA parameters with the base model for inference.

\begin{algorithm}[tb]
   \caption{FedShield-LLM}
   \label{alg:fedshield_llm}
\begin{algorithmic}
   \STATE \textbf{Input:} $N$: Number of clients, $R$: Number of rounds, $D$: Local datasets, $w_0$: Initial global model, $C$: CKKS encryption context, $p_0$: Initial pruning rate, $p_{\text{target}}$: Maximum pruning rate, $t_{\text{eff}}$: Pruning start round, $t_{\text{target}}$: Round for maximum pruning
   \STATE \textbf{Output:} $w_{T}$: Final global model, $L$: Training loss history

   \STATE Initialize global model $w_0$
   \STATE Set up CKKS encryption context $C$

   \FOR{$t = 1$ \textbf{to} $R$}
       \STATE Select a subset of clients $n_t \subseteq N$

       \STATE Compute pruning rate: $p_t \gets \max\left(0, \frac{t - t_{\text{eff}}}{t_{\text{target}} - t_{\text{eff}}}\right) \cdot \left(p_{\text{target}} - p_0\right) + p_0$

       \FOR{each client $i \in n_t$}
           \STATE Synchronize global model $w_t$ to local model $w_{i}$
           \STATE Load client dataset $D_i$
           \STATE Fine-tune $w_{i}$ using LoRA for 1 epoch on $D_i$
           \STATE Extract LoRA parameters $\Delta w_{i} = A_i \cdot B_i$

           \STATE Calculate mask: $m_{i} \gets \text{mask}(\Delta w_{i}, p_t, \text{L1 norm})$
           \STATE Sparsed model updates: $\Delta w_{i}^{p} \gets \Delta w_{i} \odot m_{i}$
           \STATE Encrypt $\Delta w_{i}^{p}$ using CKKS encryption
       \ENDFOR

       \STATE Aggregate encrypted updates: $\bar{w}_t \gets \frac{1}{|n_t|} \sum_{i \in n_t} \Delta w_{i}^{p}$
       \STATE Decrypt $\bar{w}_t$ to obtain the aggregated update
       \STATE Update global model: $w_{t+1} \gets w_t + \bar{w}_t$

       \STATE Evaluate global model and compute training loss
       \STATE Save model checkpoint periodically
   \ENDFOR

\STATE \textbf{return} $w_T$: Final global model, $L$: Training loss history
\end{algorithmic}
\end{algorithm}

To further enhance instruction-following capability, we incorporate FedIT, which enables clients to fine-tune the model using instruction-response pairs from private datasets. This ensures that the learned model remains aligned with task-specific instructions while preserving data locality and privacy~\cite{r3_ye2024openfedllm}.

Next, we describe each component of the framework in detail, starting with the LoRA-based local update mechanism, followed by pruning and secure aggregation.

\paragraph{LoRA-based Local Update}

During each communication round, the server distributes the current global model $w_t$ to selected clients. Instead of updating all model parameters, each client fine-tunes only the LoRA adapters, significantly reducing computational overhead.

The LoRA update is defined as:
\begin{align}
    \Delta w_{i} = A_{i,t} \cdot B_{i,t},
\end{align}

where $A_{i,t} \in \mathbb{R}^{d \times r}$ and $B_{i,t} \in \mathbb{R}^{r \times d}$ are low-rank matrices for the client $i$ at time $t$.. This formulation ensures that the number of trainable parameters is reduced from $O(d^2)$ to $O(rd)$, making fine-tuning feasible for distributed clients.

\paragraph{Pruning}

After computing the LoRA updates, each client applies unstructured pruning to remove less significant parameters. This step serves two purposes: (i) reducing communication overhead and (ii) limiting the exposure of sensitive information by minimizing the update surface.

The pruning mask is computed as:
\begin{align}
m_i \leftarrow \text{mask}(\Delta w_i, p_t, \text{L1 norm}),
\end{align}

where $p_t$ is the pruning rate at round $t$ \cite{mia2024quancrypt}. The mask retains parameters with larger magnitudes and zeroes smaller ones.

The pruned update is:
\begin{align}
\Delta w_i^{p} \leftarrow \Delta w_i \odot m_i,
\end{align}

where $\odot$ denotes element-wise multiplication. This process ensures that smaller, less impactful weights are set to zero, focusing computational resources on the most influential parameters.

\paragraph{Secure Aggregation via FHE}

To guarantee privacy, each client encrypts its pruned update before transmission. We use the CKKS HE scheme, which allows arithmetic operations directly on encrypted data.

\begin{align}
\Delta w_i^{p} \leftarrow \text{Enc}(\Delta w_i^{p}, \mathcal{C}),
\end{align}

where $\mathcal{C}$ is the encryption context. The server aggregates encrypted updates without accessing plaintext values.

The aggregation step is:
\begin{align}
\bar{w}_t \leftarrow \frac{1}{|n_t|} \sum_{i \in n_t} \Delta w_i^{p},
\end{align}

After aggregation, the result is decrypted to obtain the global LoRA update, which is applied to the global model and redistributed for the next round.

\noindent Overall, the proposed methodology demonstrates how secure, scalable, and efficient FL can be achieved by integrating FHE with pruning and LoRA. Experimental results validate its ability to maintain competitive performance while preserving strict privacy standards, making it suitable for sensitive applications in domains like healthcare, finance, and autonomous systems. 

\subsection{Convergence Analysis}
We provide a convergence guarantee for the FedShield-LLM training algorithm. In essence, FedShield-LLM performs an encrypted and sparsified variant of the standard FedAvg procedure \cite{mcmahan2017communication}. The added encryption does not alter any numerical updates, and the pruning step can be viewed as a form of gradient sparsification, which, under certain conditions, does not manipulate the convergence rate of SGD \cite{li2020federated}. Below we formalize the convergence result.
\newtheorem{theorem}{Theorem}
\begin{theorem}[Convergence of FedShield-LLM]
\label{thm:convergence}
Let $F(w)=\sum_{i=1}^N p_i F_i(w)$ be the global objective, where $F_i(w)$ is the local loss on client $i$ and $p_i$ is a weighting factor (e.g., $p_i = \frac{n_i}{\sum_j n_j}$ by data size). Suppose each $F_i$ is $L$-smooth and bounded below, and that the stochastic gradients have bounded variance $\sigma^2$. Assume FedShield-LLM uses (1) LoRA-based local updates $\Delta w_i^r = A_i^r B_i^r$, (2) dynamic unstructured pruning that zeros out parameters with a bounded compression error $\|e_i^r\| \le \epsilon_r$, and (3) secure aggregation via CKKS encryption. If the learning rate $\eta$ is sufficiently small, then after $R$ rounds of communication, for some constant $C>0$ (depending on $L$ and $\sigma^2$). In other words, FedShield-LLM converges sublinearly to a stationary point of $F(w)$, with an $O(1/\sqrt{R})$ rate in the general non-convex setting.
\end{theorem}

\begin{equation}
{
\min_{0 \le r < R} \mathbb{E}\!\big[\|\nabla F(w^r)\|^2\big] \;\le\; \frac{C}{\sqrt{R}}\,, 
}
\end{equation}

\begin{proof}[Proof Sketch]
Our proof builds on the classical convergence analysis of federated SGD \cite{mcmahan2017communication}  and results from compressed distributed optimization \cite{li2020federated}.

Step 1: Baseline FedAvg dynamics.
In the absence of pruning and encryption, FedShield-LLM reduces to vanilla FedAvg. In each round $r$, a subset of clients compute local LoRA updates $\Delta w_i^r = A_i^r B_i^r$ and send them to the server. The server aggregates these updates (via a weighted average $\bar{\Delta}^r = \sum_i p_i\,\Delta w_i^r$) and updates the global model: $w^{r+1} = w^r + \eta\,\bar{\Delta}^r$. Under the assumed smoothness and bounded variance conditions, this Federated SGD procedure is known to converge to a stationary point at a sublinear rate $O(1/\sqrt{R})$ for non-convex objectives \cite{yu2019parallel}. In particular, with $\eta$ small enough one has $\min_{0\le r<R}\mathbb{E}[\|\nabla F(w^r)\|^2] = O(1/\sqrt{R})$ \cite{reddi2020adaptive}, which matches the typical convergence rate of centralized SGD in the non-convex setting \cite{ghadimi2013stochastic}.

Step 2: Effect of encryption. FedShield-LLM uses CKKS HE for secure aggregation of updates. Since homomorphic addition yields exactly the same sum as plaintext (just encrypted), the aggregated update $\bar{\Delta}^r$ is numerically identical to that of vanilla FedAvg \cite{cheon2017homomorphic}. Therefore, encryption has no effect on the optimization dynamics or convergence, it merely hides individual $\Delta w_i^r$ from the server.

Step 3: Impact of pruning. After computing $\Delta w_i^r$, each client applies an unstructured pruning mask that zeros out the smallest-magnitude entries (according to the current prune rate $p_t$), producing a sparse update $\Delta w_{i,\text{prune}}^r = \Delta w_i^r \circ m_i^r$. Let $e_i^r := \Delta w_i^r - \Delta w_{i,\text{prune}}^r$ be the pruning error (the dropped components) \cite{alistarh2017qsgd}. The server then aggregates the pruned updates:

\begin{equation}
\begin{aligned}
\bar{\Delta}^r 
&= \sum_{i} p_i\,\Delta w_{i,\text{prune}}^r  \\
&= \sum_{i} p_i\,\Delta w_i^r \;-\; \sum_{i} p_i\,e_i^r \\
&= \bar{\Delta}_\text{full}^r \;-\; \bar{e}^r \,.
\end{aligned}
\label{eq:pruned-update}
\end{equation}

where $\bar{\Delta}_\text{full}^r$ is the full (unpruned) average update and $\bar{e}^r := \sum_i p_i e_i^r$ is the aggregated pruning error \cite{alistarh2018convergence}. Thus, pruning effectively adds a perturbation $-\bar{e}^r$ to the update at each round.

Step 4: Descent with smoothness. By $L$-smoothness of $F$, a standard descent lemma gives:

\begin{equation}
\begin{aligned}
F(w^{r+1}) 
&\;\le\; F(w^r) 
    + \eta\,\langle \nabla F(w^r),\, \bar{\Delta}^r \rangle  \\
&\quad + \frac{L\,\eta^2}{2}\,\|\bar{\Delta}^r\|^2 \,.
\end{aligned}
\label{eq:smoothness}
\end{equation}

Substituting $\bar{\Delta}^r = \bar{\Delta}_\text{full}^r - \bar{e}^r$ and taking expectation, we obtain an approximate descent inequality:

\begin{equation}
\begin{aligned}
\mathbb{E}\!\left[F(w^{r+1}) - F(w^r)\right]
&\;\le\; -\eta\,\mathbb{E}\!\left[\|\nabla F(w^r)\|^2\right]  \\
&\quad +\; \eta\,\mathbb{E}\!\left[\left\langle \nabla F(w^r),\, -\bar{e}^r \right\rangle\right] \\
&\quad +\; \frac{L\eta^2}{2}\,\mathbb{E}\!\left[\|\bar{\Delta}^r\|^2\right]\,.
\end{aligned}
\label{eq:expected-bound}
\end{equation}

The first term represents descent due to the true gradient, while the second term is an ascent error due to pruning. By Cauchy-Schwarz, $\langle \nabla F(w^r), -\bar{e}^r \rangle \le \|\nabla F(w^r)\|\,\|\bar{e}^r\|$. Using bounded variance and assuming $\|\bar{e}^r\| \le \epsilon_r$ (by our pruning design), the error term will be of order $\eta\,\|\nabla F(w^r)\|\,\epsilon_r$. For a moderate pruning ratio, this bias remains small \cite{wangni2018gradient}, and can be further mitigated by decreasing $\eta$ as needed. Meanwhile, the third term $L\eta^2 \|\bar{\Delta}^r\|^2$ is $O(\eta^2)$ and can be made negligible for small $\eta$.

Step 5: Sublinear convergence. Summing the above inequality over $r=0$ to $R-1$ and telescoping the left-hand side yields a bound on $F(w^R) - F(w^0)$. Since $F$ is bounded below, the total descent is finite. Thus, rearranging terms and dividing by $R$, we obtain:

\begin{equation}
\begin{aligned}
\frac{1}{R}\sum_{r=0}^{R-1}\mathbb{E}\!\left[\|\nabla F(w^r)\|^2\right]
&\;\le\; \frac{2\,\big(F(w^0) - F(w^*)\big)}{\eta\,R}  \\
&\quad +\; \frac{2\,L\,\eta}{R}\sum_{r=0}^{R-1}\mathbb{E}\!\left[\|\bar{\Delta}^r\|^2\right]  \\
&\quad +\; \frac{2\,\eta}{R}\sum_{r=0}^{R-1}\|\nabla F(w^r)\|\,\epsilon_r\,.
\end{aligned}
\label{eq:convergence-bound}
\end{equation}

where $F(w^*)$ is the lower bound of $F$. Under the bounded variance assumption, $\mathbb{E}[\|\bar{\Delta}^r\|^2]$ remains bounded by a constant related to $\sigma^2$. Moreover, if the pruning schedule is gentle (increasing sparsity slowly) such that $\epsilon_r$ does not grow too large, the last term can be controlled \cite{lin2017deep}. Therefore, for a sufficiently small $\eta$, the dominant term is $O(1/(\eta R))$, which leads to the bound $\min_{0\le r < R}\mathbb{E}[\|\nabla F(w^r)\|^2] \le O(1/\sqrt{R})$ as claimed. This indicates that FedShield-LLM reaches an $\epsilon$-stationary point (where $\mathbb{E}\|\nabla F(w)\|^2 \le \epsilon$) in $O(1/\epsilon^2)$ communication rounds, which is the same order of complexity as standard FedAvg without pruning or encryption \cite{khaled2020tighter}.
\end{proof}

\theoremstyle{remark}
\newtheorem{remark}{Remark}
\begin{remark}
The above convergence guarantee is provided for general non-convex objectives, which is appropriate since fine-tuning deep networks (LLMs) is a non-convex optimization problem. In this setting, a sublinear $O(1/\sqrt{R})$ rate is typical \cite{bottou2018optimization}. If additional structure is present—e.g. if $F(w)$ satisfies the Polyak-Łojasiewicz or $\mu$-strong convexity condition—then faster convergence can be achieved (e.g. a linear rate) \cite{karimi2016linear}. In practice, our experiments confirm that FedShield-LLM converges stably and rapidly, closely matching the accuracy trajectory of Vanilla FedAvg in early rounds. The encryption module ensures updates are aggregated securely without affecting convergence, and the gradual pruning schedule preserves model quality while still yielding convergence, demonstrating the algorithm’s stability under both encryption and pruning.
\end{remark}

\subsection{Security and Robustness Analysis of FedShield-LLM}

\begin{theorem}[Security and Robustness of FedShield-LLM]
\label{thm:fedshield-security}
Assuming the semantic security of CKKS under the RLWE assumption and given that each client's LoRA update is sparsified via unstructured pruning with rate $p_t$, the FedShield-LLM protocol protects against inference attacks (e.g., gradient inversion) from both passive adversaries and colluding servers, with negligible advantage for any polynomial-time attacker.
\end{theorem}

\begin{proof}
\textbf{(1) Semantic Security via CKKS:}  
Each client encrypts its sparsified LoRA update $\Delta w_i^p = \Delta A_i \cdot \Delta B_i$ using the CKKS encryption scheme before transmission. CKKS operates over the polynomial ring $R_q = \mathbb{Z}_q[x]/(x^N + 1)$, with a secret key $s \leftarrow \chi$ and a public key $\mathsf{pk} = (b, a)$ where $b = -a \cdot s + e \pmod{q}$. The message is first scaled and encoded, and then encrypted as:
\[
c_0 = b \cdot u + e_1 + m', \quad c_1 = a \cdot u + e_2
\]
The server aggregates encrypted client updates homomorphically and holds only the public key. A trusted party with the secret key performs decryption of the final aggregated ciphertext. Under the RLWE assumption, CKKS ensures IND-CPA security:

\begin{equation}
\begin{aligned}
\text{Enc}(\Delta) \approx_c \text{Enc}(\Delta') 
\quad &\Rightarrow \\
\Pr[\mathcal{A} \text{ distinguishes}] 
&\leq \text{negl}(\lambda)
\end{aligned}
\end{equation}

Thus, ciphertexts reveal no meaningful information to adversaries, including honest-but-curious servers.

\textbf{(2) Robustness from Sparsified Aggregation:}  
Each client applies a binary pruning mask $m_i \in \{0,1\}^d$ to their LoRA update, yielding:
\begin{align}
\Delta w_i^{p} = m_i \odot \Delta w_i, \quad \text{with } \|m_i\|_0 = (1 - p_t)d
\end{align}
After homomorphic aggregation and decryption, the server observes only the combined sparse update:
\begin{align}
{\Delta w}_{\text{agg}} = \sum_{i \in \mathcal{C}_t} m_i \odot \Delta w_i^{p}
\end{align}
This results in an underdetermined system, where:
\begin{align*}
\exists\, \{\Delta w_1', ..., \Delta w_N'\} &\neq \{\Delta w_1, ..., \Delta w_N\}:\\
&\sum_i m_i \odot \Delta w_i' = {\Delta w}_{\text{agg}}
\end{align*}
Due to:
\begin{itemize}
    \item \textbf{Sparsity}, and
    \item \textbf{Large client set size} ($N \gg 1$),
\end{itemize}
the inversion becomes ill-posed. Even in the worst-case collusion scenario (server and $N-1$ clients), the remaining client's update is both sparse and low-rank due to LoRA, significantly limiting reconstructability.

Thus, FedShield-LLM integrates HE with secure key separation and sparsified LoRA updates to achieve a two-layered defense against gradient inversion and data inference attacks, ensuring robustness under standard cryptographic assumptions.
\end{proof}

\subsection{Computational Complexity Analysis}
We analyze the computational and communication complexity of each component in FedShield-LLM: LoRA fine-tuning, HE with secure aggregation, unstructured pruning, and communication overhead.

\subsubsection{LoRA Fine-Tuning} Let $P$ be the total model parameters and $P_{\text{LoRA}} \ll P$ the LoRA trainable parameters. Each client trains only $P_{\text{LoRA}} = O(r \cdot d)$ parameters (rank $r$, hidden size $d$). The per-iteration time complexity is $O(P + P_{\text{LoRA}})$, dominated by forward/backward over $P$. Optimizer state and memory usage are only for $P_{\text{LoRA}}$, yielding substantial speed-up and efficiency. Compared to full fine-tuning, LoRA reduces memory and compute significantly, enabling client-side feasibility for cross-silo FL.

\subsubsection{FHE and Aggregation} 
Each client encrypts its LoRA update using CKKS. Encryption complexity per client is $O(n_c \cdot N_{\text{poly}} \log N_{\text{poly}})$, where $n_c$ is the number of ciphertexts. With vector packing, updates (e.g., 30M parameters) are reduced to a few thousand ciphertexts. Encryption takes $\sim$15 seconds per client; decryption of aggregated updates is $<1$ second. Homomorphic addition on the server is $O(N_{\text{poly}})$ per ciphertext. Compared to MPC-based schemes~\cite{bonawitz2017practical,singh2022federated}, FHE trades higher compute for simpler one-round communication and supports post-aggregation operations.

\subsubsection{Unstructured Pruning} 
Pruning selects a fraction $p_t \in [0.2, 0.7]$ of smallest magnitude parameters. Threshold selection via partial sort is $O(P_{\text{LoRA}})$; zeroing is linear. This introduces negligible runtime and can reduce encryption cost. Since pruning is done \emph{after} gradient computation, it does not affect training dynamics.

\subsubsection{Communication Overhead} 
LoRA reduces upload size from $P$ to $P_{\text{LoRA}}$ parameters (e.g., from 1200MB to 120MB). Pruning further reduces computational overhead. Encrypted updates are acceptable in cross-silo FL with high-bandwidth links. While CKKS ciphertexts inflate size, fewer communication rounds and sparse updates reduce overall cost. In contrast, MPC-based secure aggregation incurs lower computational overhead but adds protocol complexity.

%% file: Experiments_and_Result_Analysis.tex
\section{Experiments and Result Analysis} \label{sec:experiments}

\subsection{Experimental Setup}
Our experiments were conducted on an Ubuntu server equipped with two NVIDIA RTX A6000 GPUs, an Intel Core i9 processor, and 128 GB of RAM. The study explored FL under Independent and Identically Distributed (IID) data scenarios, simulating 3 clients per communication round. Each client performed local training using the Adam optimizer with a learning rate of $5 \times 10^{-5}$, a batch size of 16, and gradient accumulation steps set to 1. The sequence length was fixed at 512, and each round consisted of a single local epoch per client. The pruning rate was progressively increased from $20\%$, achieving the target value of $50\%$ by round 200. Pruning masks were applied to generate a sparsified model by setting less important weights to zero based on the pruning rate. This approach reduced computational overhead while maintaining accuracy. To ensure robust data privacy and security, we implemented HE from TenSEAL \cite{benaissa2021tenseal}, a library built on Microsoft SEAL. The CKKS encryption scheme was configured with a polynomial modulus degree of 16384 and coefficient modulus sizes $[60, 40, 40, 40, 60]$. This setup enabled secure aggregation of model updates, sharing only the public key among clients and the server, thereby maintaining the confidentiality of individual client data. Model evaluation was conducted by decrypting the aggregated model weights at the server using the private key when necessary to assess performance. For fine-tuning efficiency, we leveraged PEFT techniques, specifically LoRA, with rank $r=32$ and alpha $\alpha=64$. This configuration optimized communication and computation, ensuring scalability in large-scale federated settings.

\subsection{Dataset and Model}

Our experiments were conducted using four diverse datasets obtained from Hugging Face: \textit{vicgalle/alpaca-gpt4} (52,002 instruction-response pairs for general-purpose fine-tuning) \cite{taori2023stanford}, \textit{FinGPT/fingpt-sentiment-train} (76,772 labeled examples for financial sentiment analysis) \cite{liu2023fingpt}, \textit{TIGER-Lab/MathInstruct} (262,039 tasks for mathematical reasoning) \cite{yue2023mammoth}, and \textit{medalpaca/medical\_meadow\_medical\_flashcards} (33,955 medical flashcard entries) \cite{han2023medalpaca}. These datasets enabled evaluation across a wide range of domains, including general, financial, instructional, and medical tasks. To ensure fairness, data was distributed across clients using an IID strategy. This involved shuffling each dataset and partitioning it uniformly so that each client received an equal portion of the data, ensuring all clients handled representative shards. 

For the model, we used \textit{meta-llama/Llama-2-7b-hf} and \textit{meta-llama/Llama-2-13b-hf}, two transformer-based pre-trained LLMs \cite{touvron2023llama} known for their strong capabilities in natural language understanding and generation tasks. This setup enabled a comprehensive evaluation of our FL framework, demonstrating its effectiveness in securely and efficiently fine-tuning large language models under uniform data distribution across clients.

\subsection{Result Analysis}

In this section, we present a comprehensive evaluation of the proposed FedShield-LLM framework, including training loss curves, text-generation outputs, and BERT-based similarity scores to assess generation quality. Our primary experiments use \textit{Llama-2-7b-hf} and \textit{Llama-2-13b-hf} fine-tuned on the \textit{fingpt-sentiment-train} and \textit{alpaca-gpt4} datasets, with comparisons against Vanilla-FL and DP-LoRA based on consistent loss and generation metrics. Moreover, we evaluate the impact of unstructured pruning on both model utility and security. Our results show that moderate pruning preserves accuracy while substantially reducing vulnerability to gradient inversion attacks. Based on these findings, we recommend a pruning rate that provides a strong balance between performance and privacy. Additional training-loss plots, qualitative generation samples, and extended comparisons with other baselines are included in the supplementary material.

\begin{figure*}[t]
  \centering

  \begin{minipage}[b]{0.32\linewidth}
    \centering
    \includegraphics[width=\linewidth]{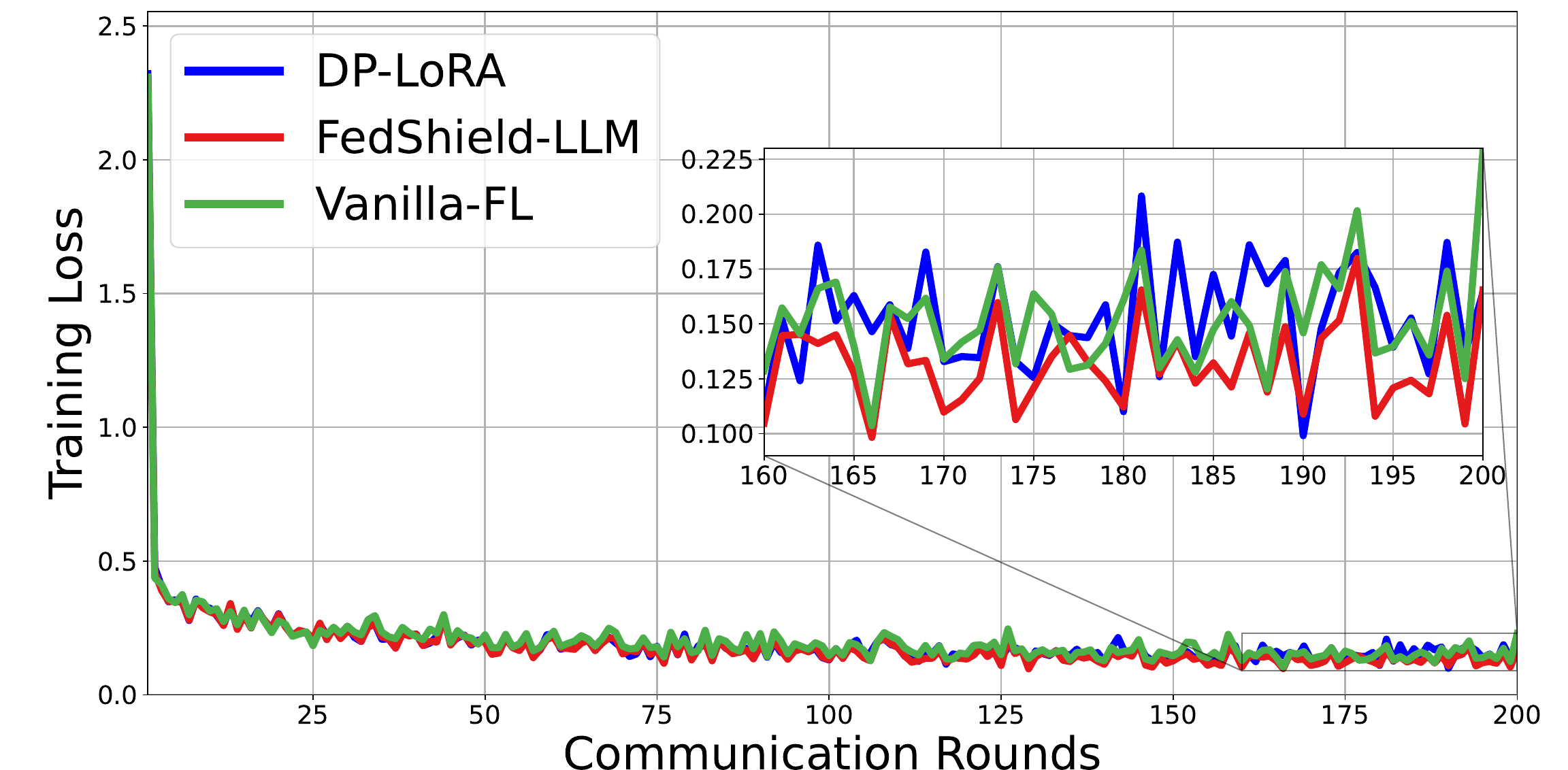}
    \subcaption{$(a)$}
    \label{fig:cl1_fin2}
  \end{minipage}
  \hfill
  \begin{minipage}[b]{0.32\linewidth}
    \centering
    \includegraphics[width=\linewidth]{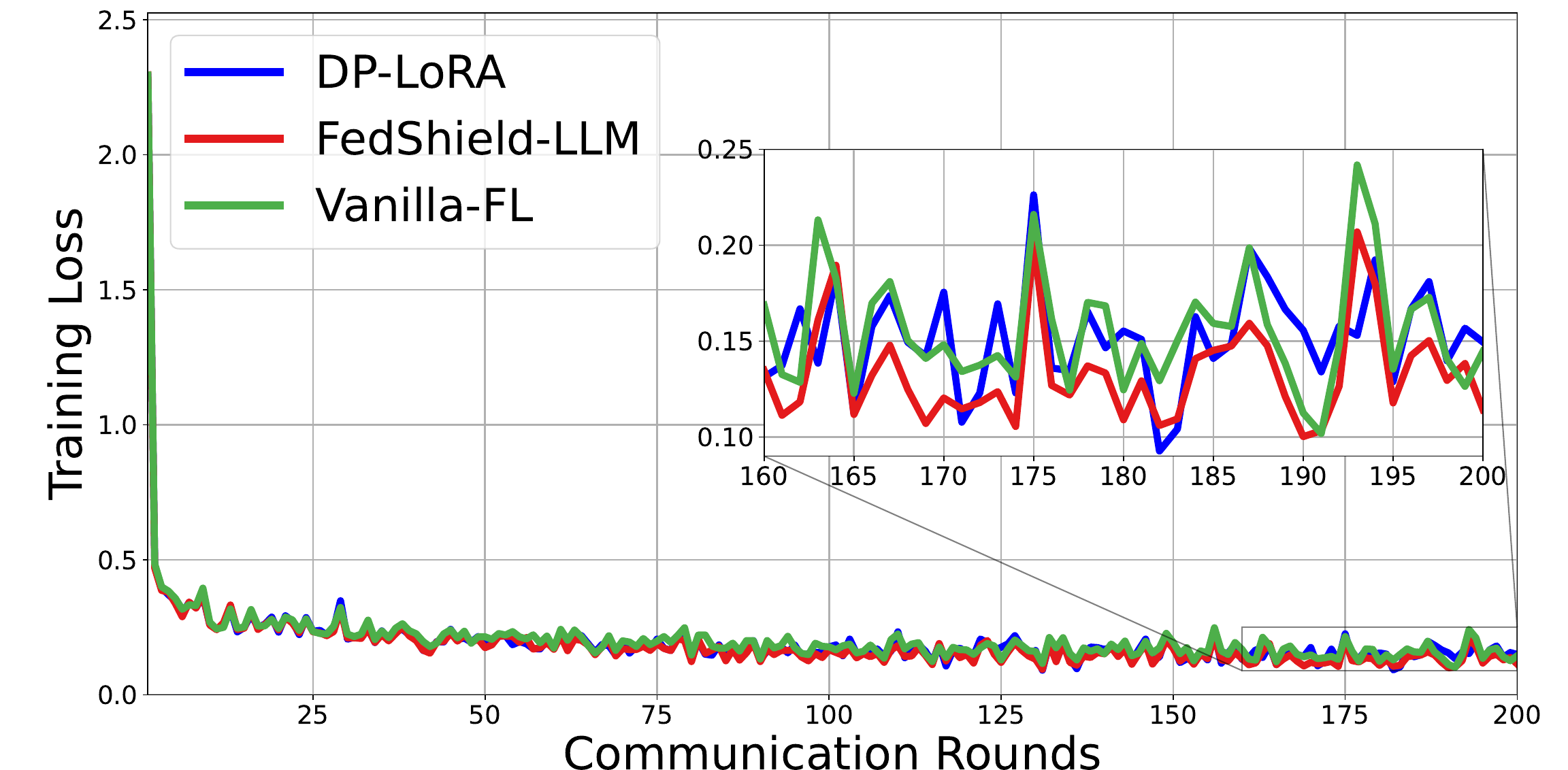}
    \subcaption{$(b)$}
    \label{fig:cl2_fin}
  \end{minipage}
  \hfill
  \begin{minipage}[b]{0.32\linewidth}
    \centering
    \includegraphics[width=\linewidth]{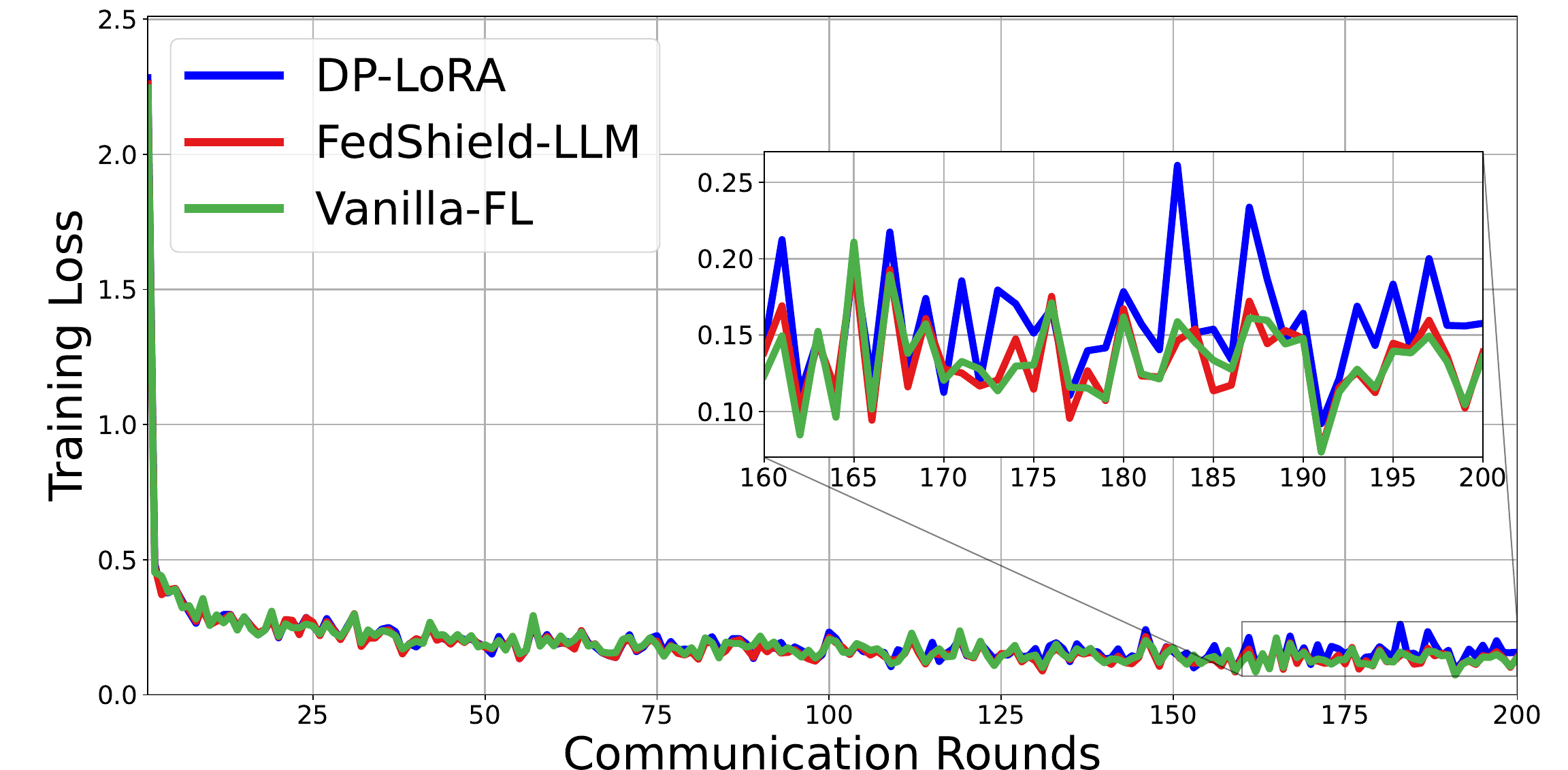}
    \subcaption{$(c)$}
    \label{fig:cl3_fin}
  \end{minipage}

  \vspace{12pt}

  \begin{minipage}[b]{0.32\linewidth}
    \centering
    \includegraphics[width=\linewidth]{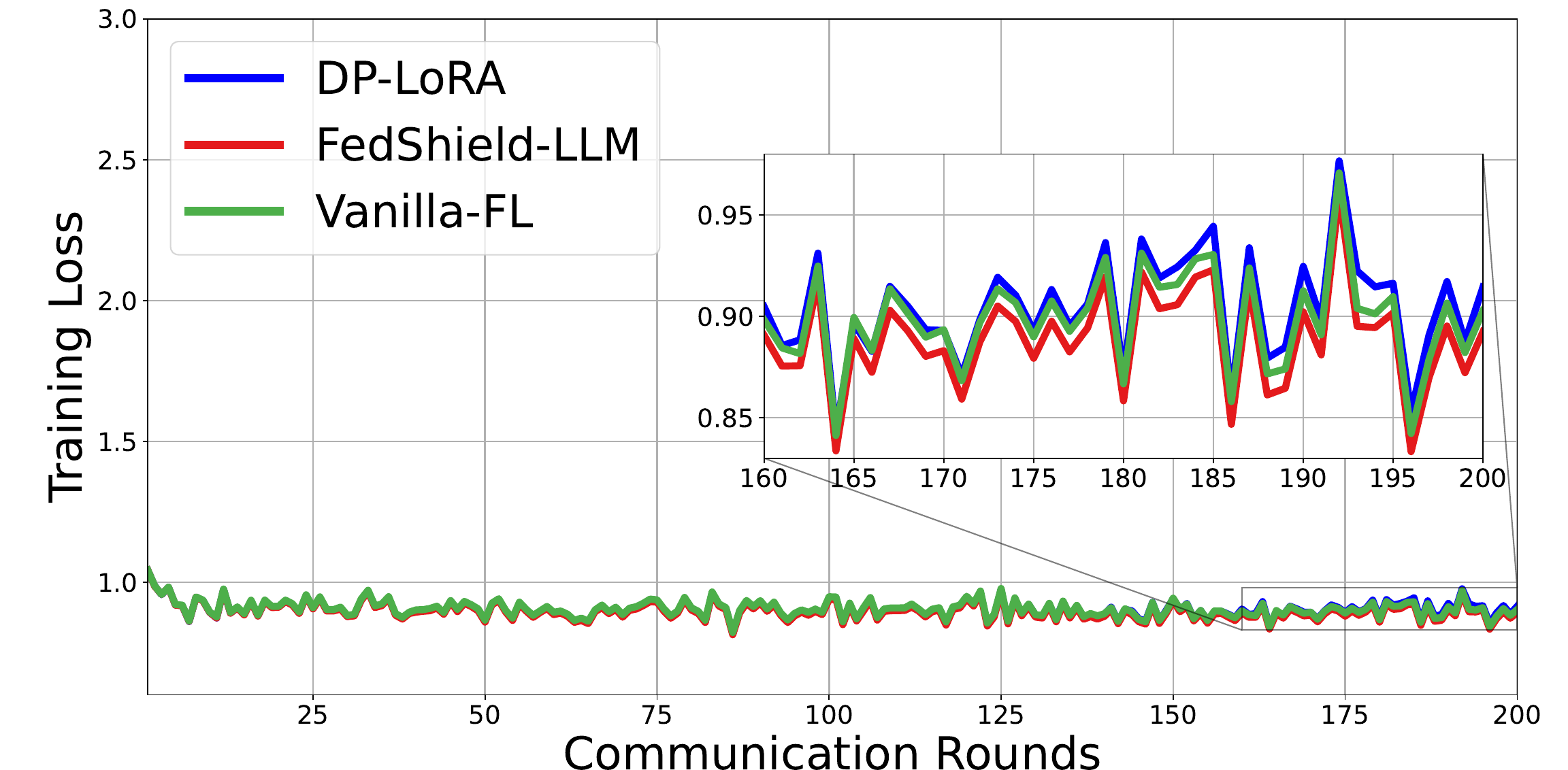}
    \subcaption{$(d)$}
    \label{fig:cl1_alpaca}
  \end{minipage}
  \hfill
  \begin{minipage}[b]{0.32\linewidth}
    \centering
    \includegraphics[width=\linewidth]{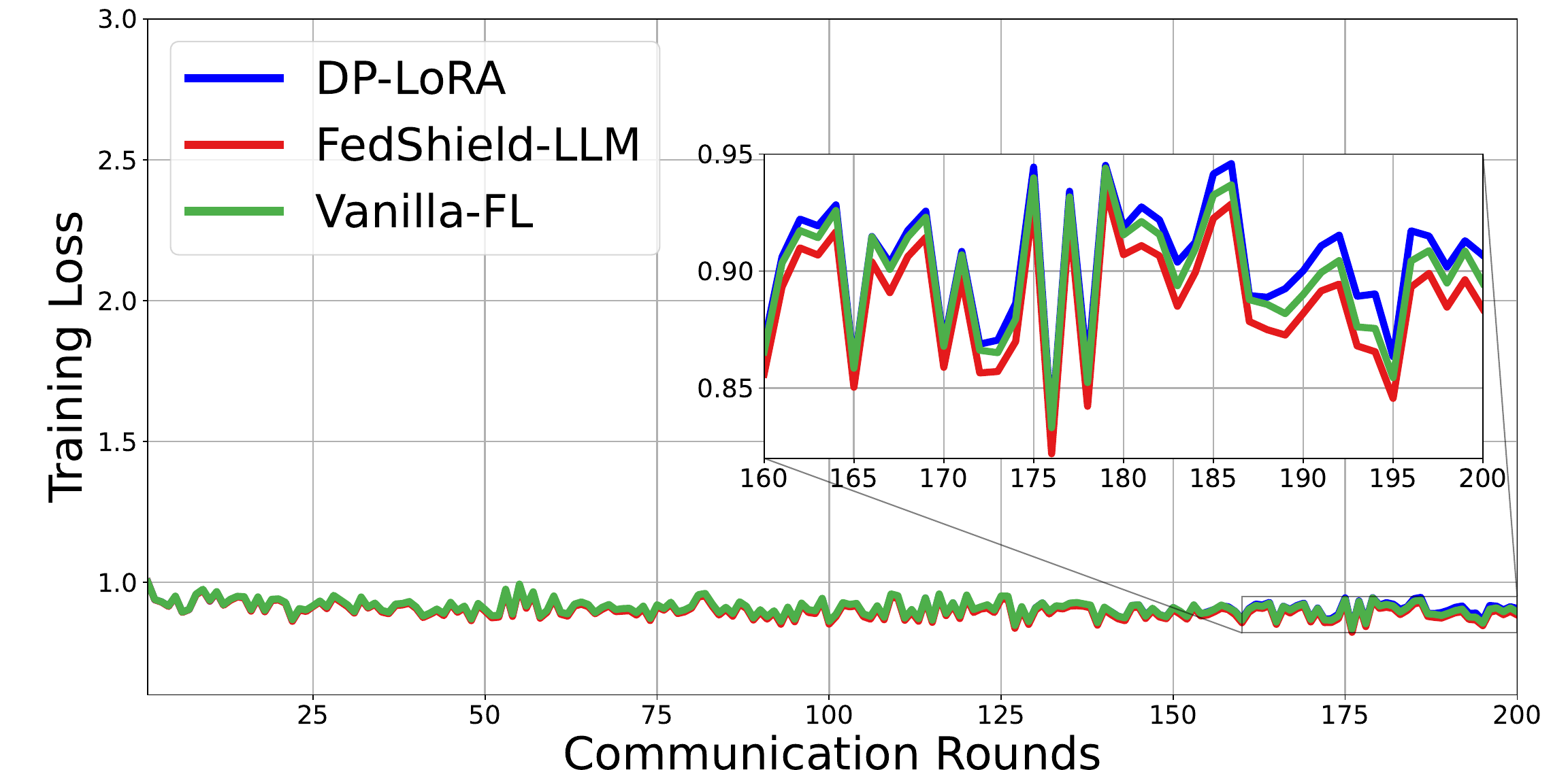}
    \subcaption{$(e)$}
    \label{fig:cl2_alpaca}
  \end{minipage}
  \hfill
  \begin{minipage}[b]{0.32\linewidth}
    \centering
    \includegraphics[width=\linewidth]{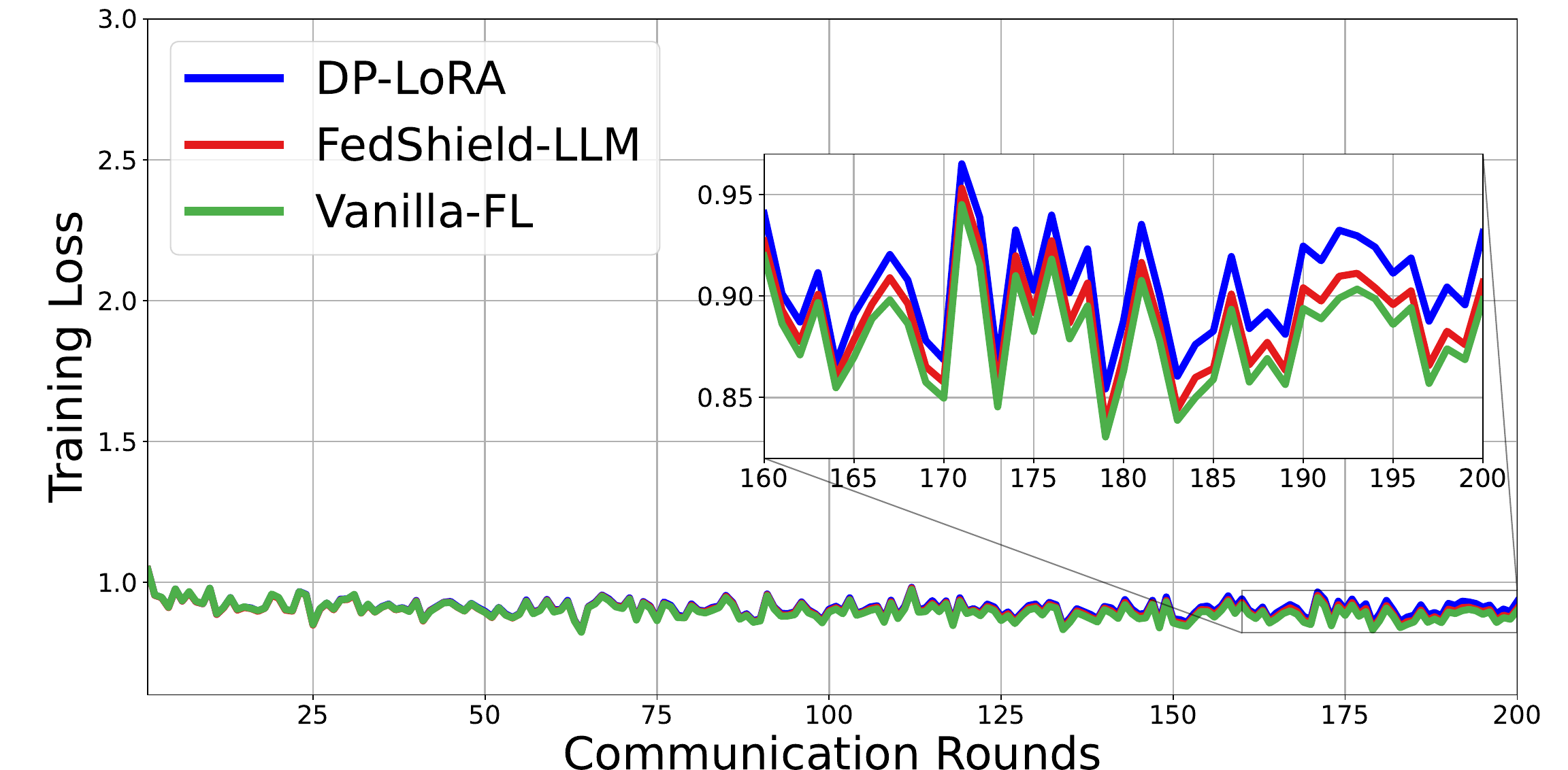}
    \subcaption{$(f)$}
    \label{fig:cl3_alpaca}
  \end{minipage}

  \caption{Comparison of training loss for three clients during federated LLM fine-tuning. Base Model: \texttt{meta-llama/Llama-2-7b-hf}. Top row (subfigures a--c): \textit{fingpt-sentiment-train}. Bottom row (subfigures d--f): \textit{vicgalle/alpaca-gpt4}.}
  \label{fig_3}
\end{figure*}

Figure~\ref{fig_3} presents the training loss patterns for three clients during federated fine-tuning on both the \textit{fingpt-sentiment-train} dataset (subfigures a–c) and the \textit{alpaca-gpt4} dataset (subfigures d–f). Across all clients and both datasets, FedShield-LLM consistently achieves lower and more stable training loss compared to Vanilla FL and DP-LoRA, demonstrating faster convergence, especially in the early communication rounds. DP-LoRA exhibits noticeably higher and more fluctuating loss throughout training, highlighting its reduced optimization efficiency in the federated setting. The uniform and stable loss trajectories observed across all clients further illustrate the robustness, scalability, and effectiveness of FedShield-LLM for privacy-preserving LLM fine-tuning under heterogeneous data distributions.

\begin{figure*}[t]
  \centering

  \begin{minipage}[b]{0.32\linewidth}
    \centering
    \includegraphics[width=\linewidth]{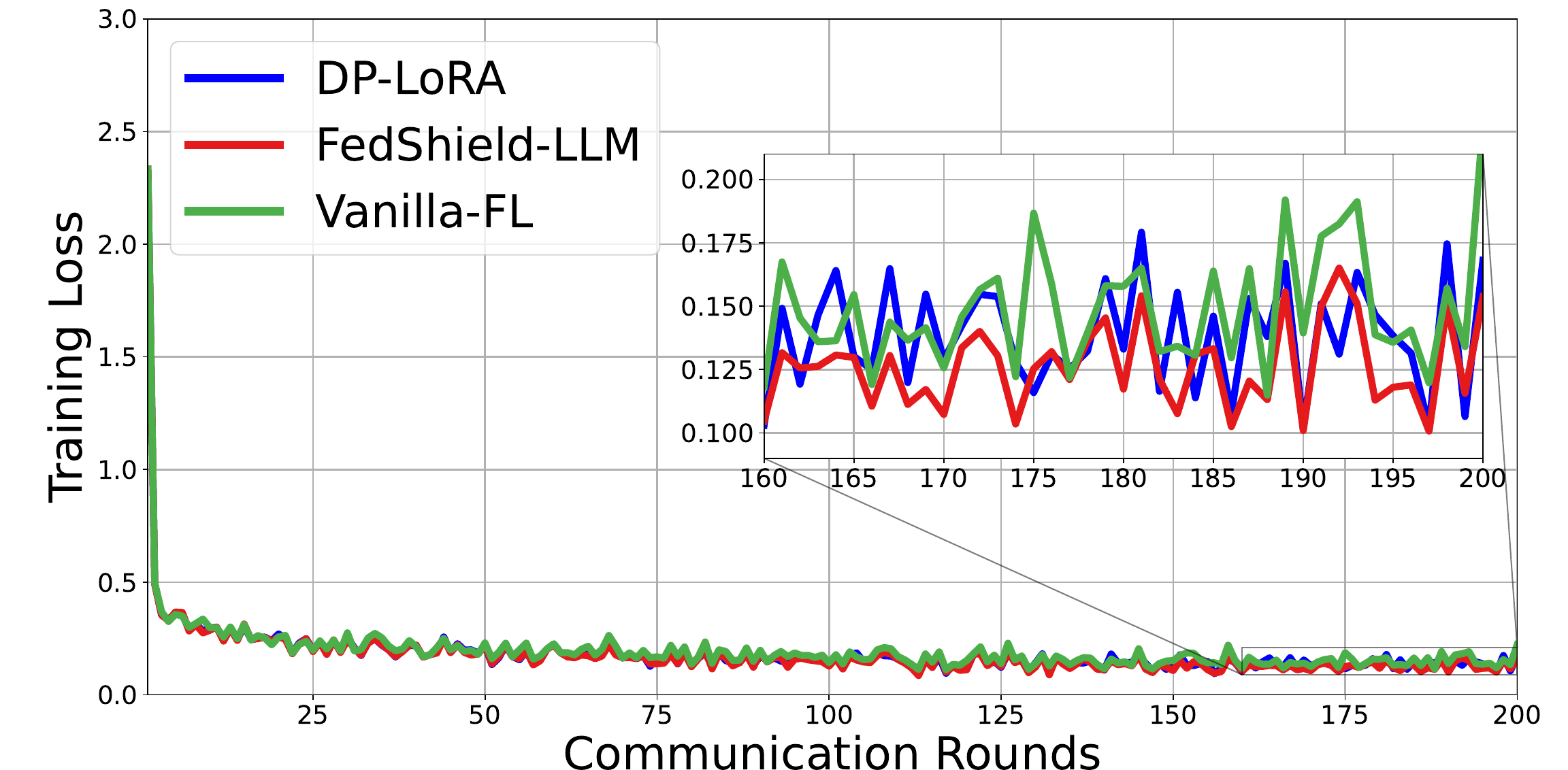}
    \subcaption{$(a)$}
    \label{fig:client1_fin}
  \end{minipage}
  \hfill
  \begin{minipage}[b]{0.32\linewidth}
    \centering
    \includegraphics[width=\linewidth]{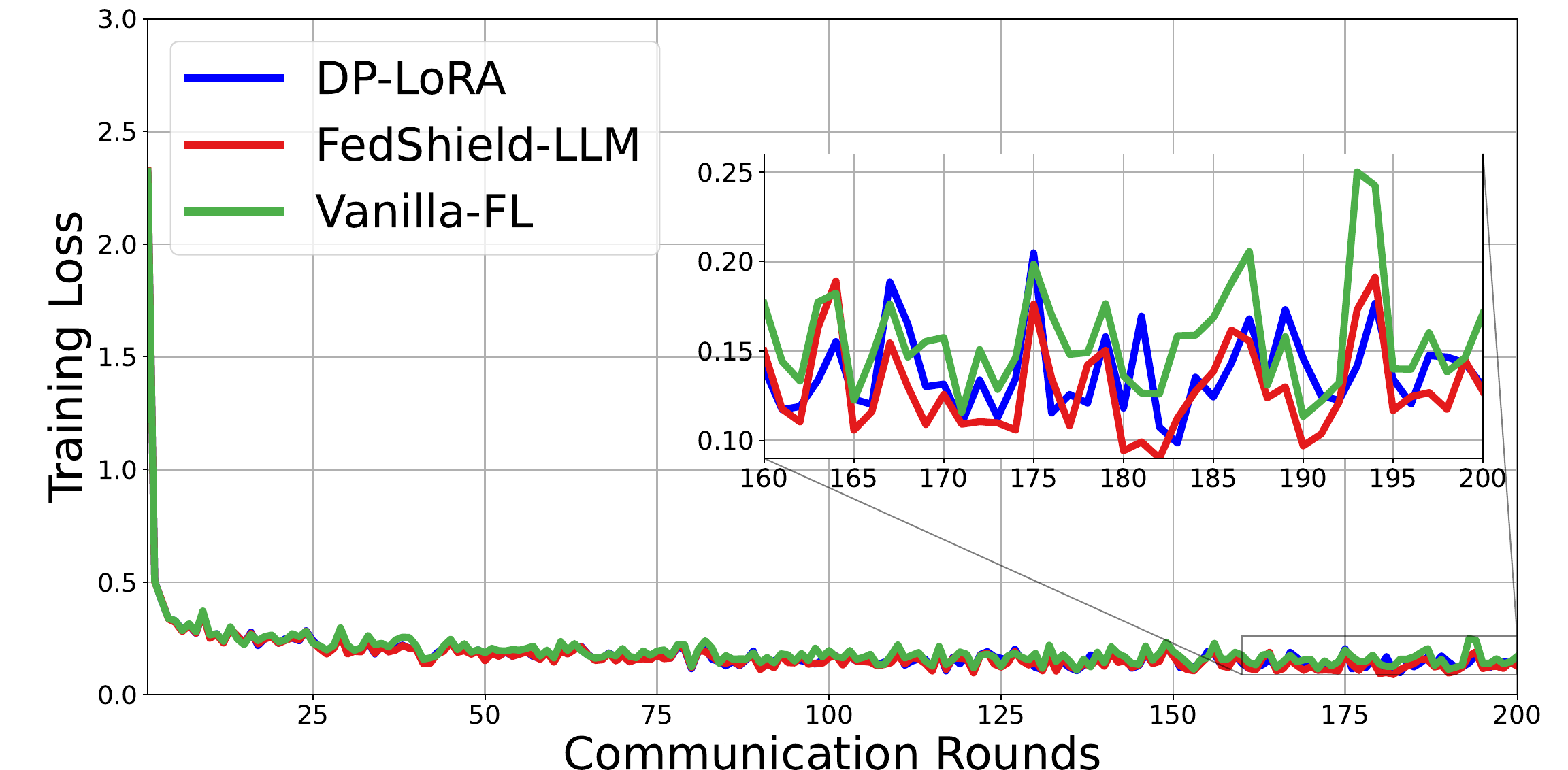}
    \subcaption{$(b)$}
    \label{fig:client2_fin}
  \end{minipage}
  \hfill
  \begin{minipage}[b]{0.32\linewidth}
    \centering
    \includegraphics[width=\linewidth]{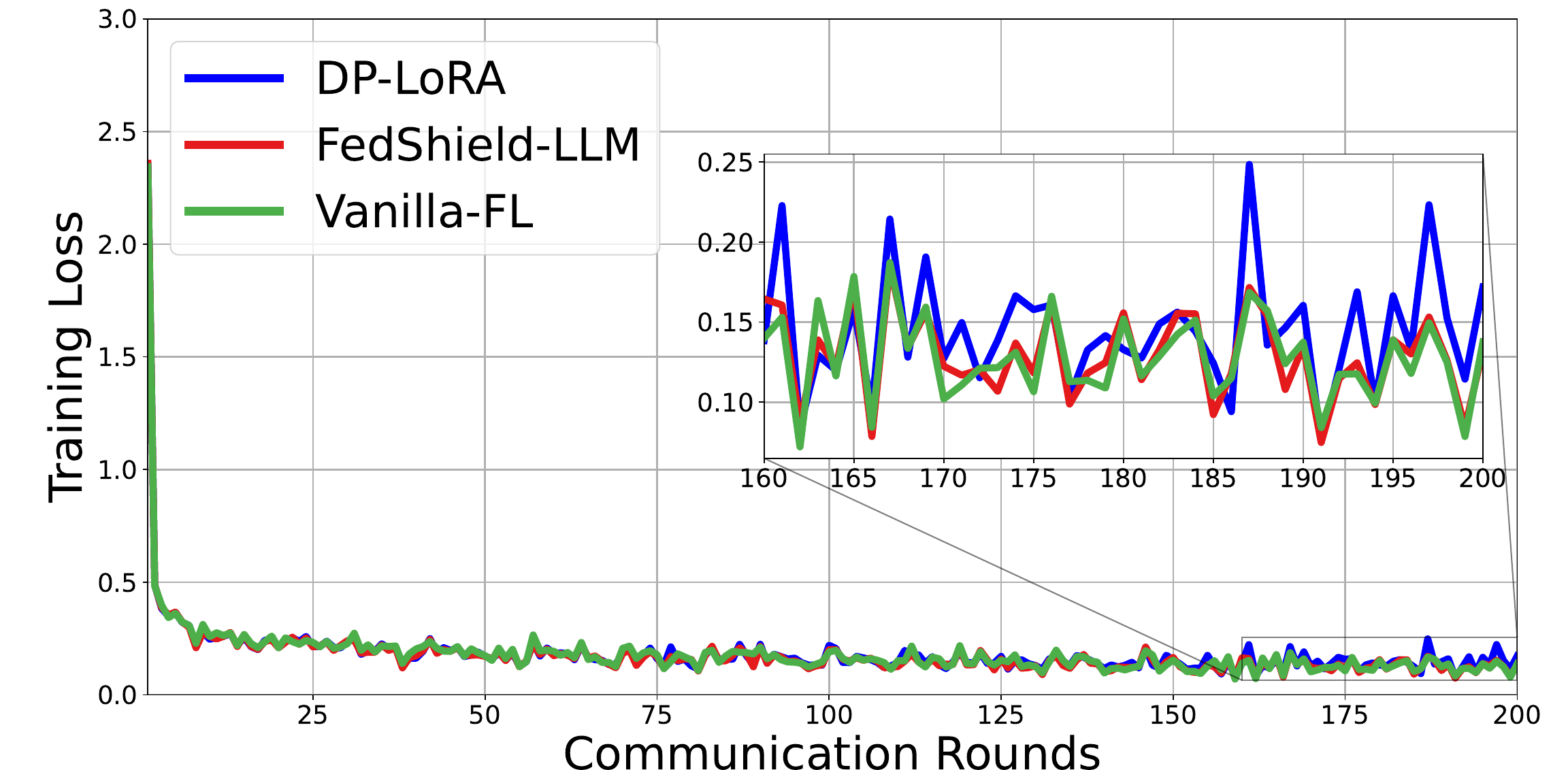}
    \subcaption{$(c)$}
    \label{fig:client3_fin}
  \end{minipage}

  \vspace{12pt}

  \begin{minipage}[b]{0.32\linewidth}
    \centering
    \includegraphics[width=\linewidth]{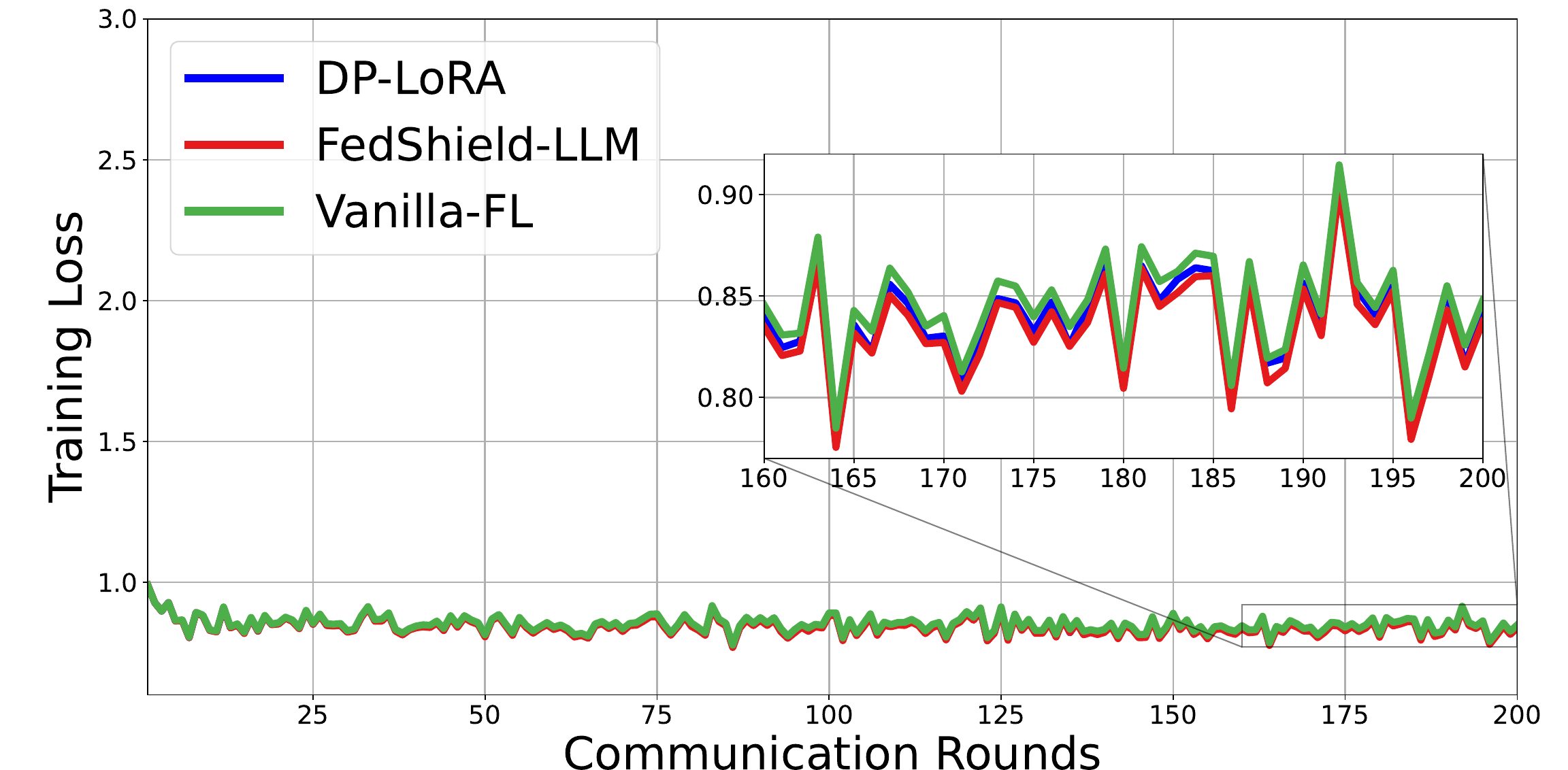}
    \subcaption{$(d)$}
    \label{fig:client1_alpaca}
  \end{minipage}
  \hfill
  \begin{minipage}[b]{0.32\linewidth}
    \centering
    \includegraphics[width=\linewidth]{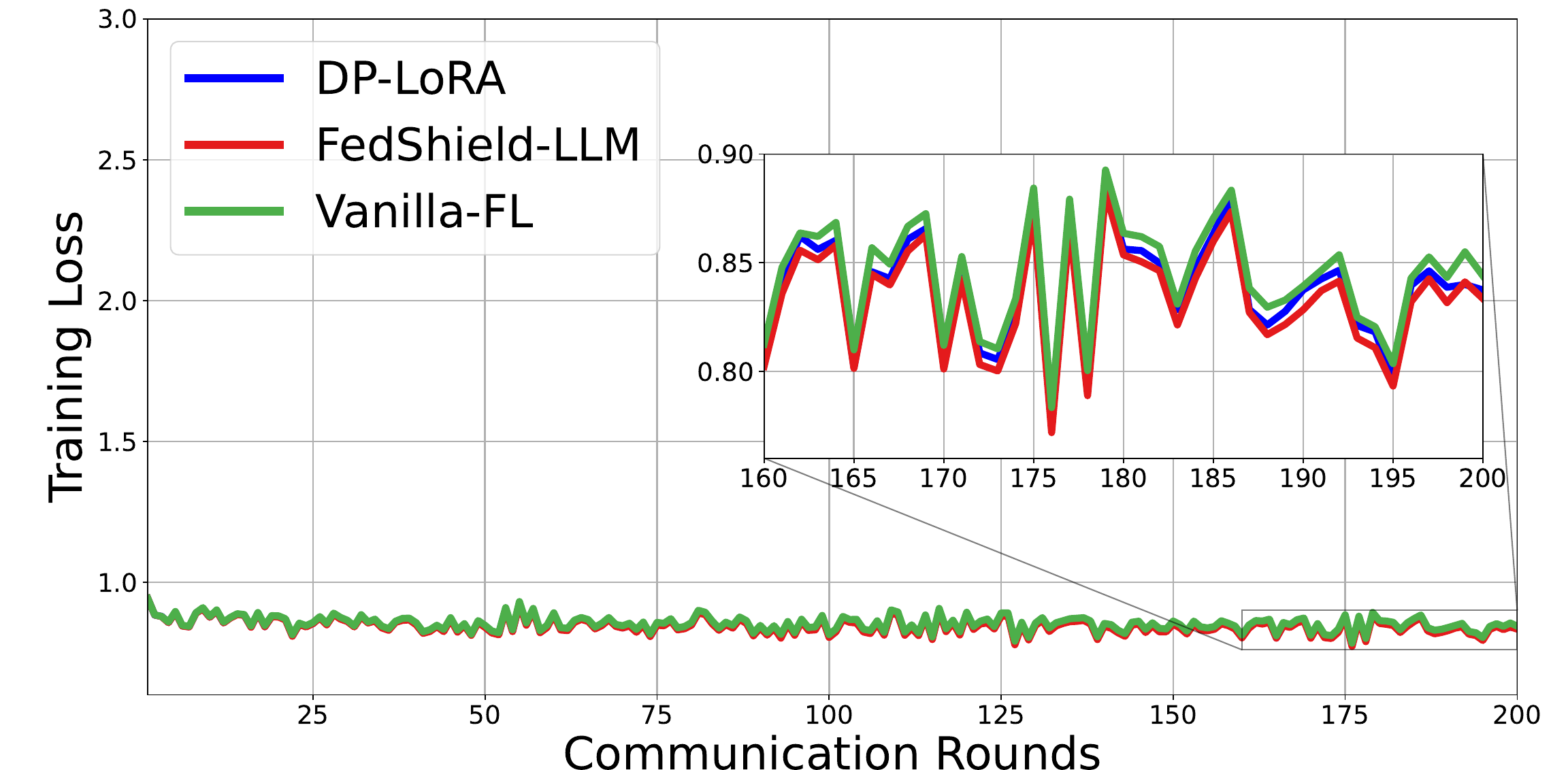}
    \subcaption{$(e)$}
    \label{fig:client2_alpaca}
  \end{minipage}
  \hfill
  \begin{minipage}[b]{0.32\linewidth}
    \centering
    \includegraphics[width=\linewidth]{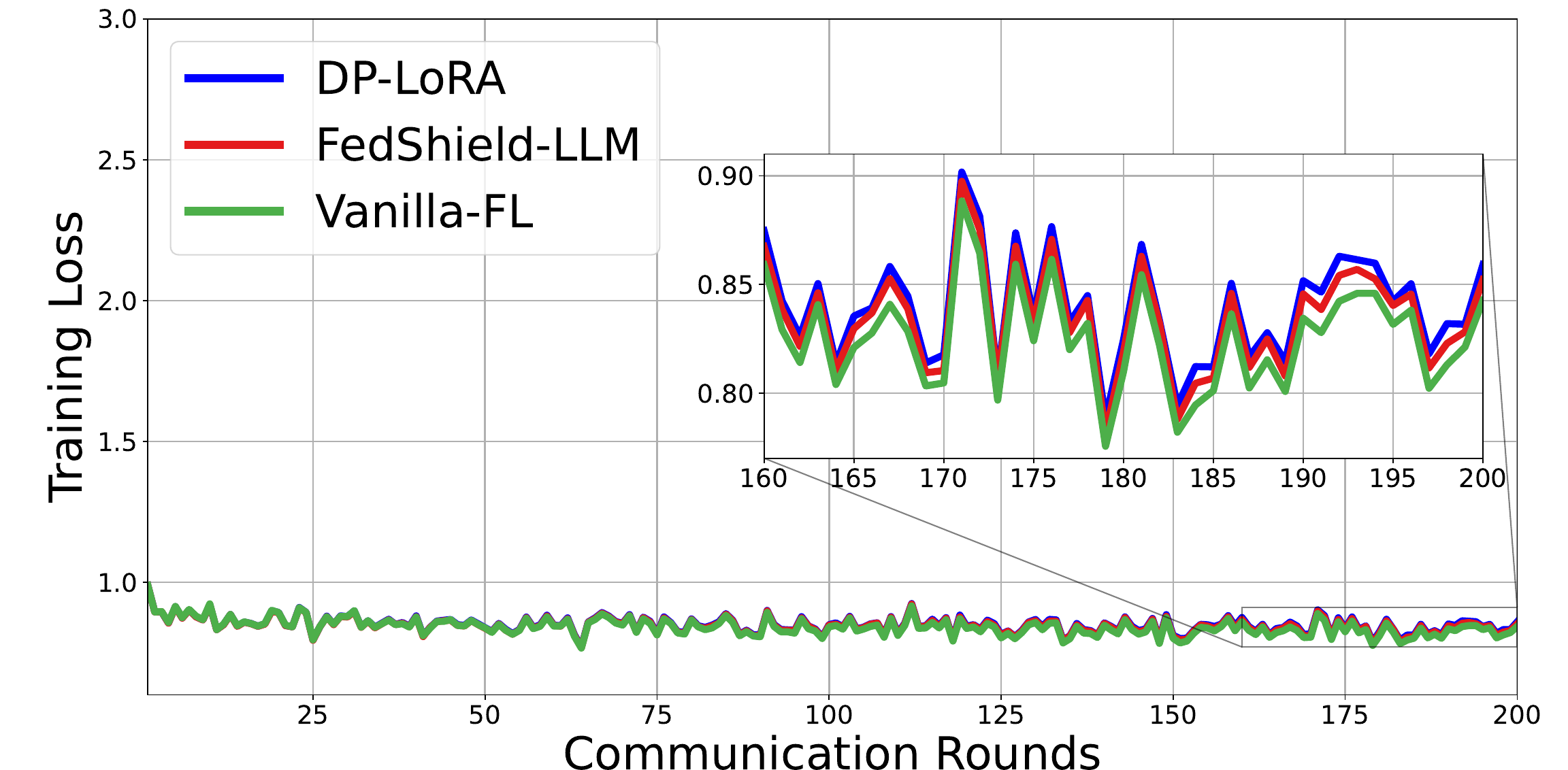}
    \subcaption{$(f)$}
    \label{fig:client3_alpaca}
  \end{minipage}

  \caption{Comparison of training loss for three clients during federated LLM fine-tuning. Base Model: \texttt{meta-llama/Llama-2-13b-hf}. Top row (subfigures a--c): \textit{fingpt-sentiment-train}. Bottom row (subfigures d--f): \textit{vicgalle/alpaca-gpt4}.}
  \label{fig_4}
\end{figure*}

Figure~\ref{fig_4} presents a comparative analysis of training loss across three clients during the federated fine-tuning of the LLaMA-2-13B model on two benchmark datasets: fingpt-sentiment-train (subfigures a–c) and vicgalle/alpaca-gpt4 (subfigures d–f). The proposed FedShield-LLM consistently achieves lower training loss compared to Vanilla FL and DP-LoRA across all communication rounds and clients. For both datasets, FedShield-LLM demonstrates improved convergence behavior and reduced variance, indicating enhanced training stability and robustness under heterogeneous client data distributions. Notably, DP-LoRA exhibits higher and more fluctuating loss values, reflecting its vulnerability to optimization inefficiencies in federated setups. The consistent performance of FedShield-LLM across diverse datasets and clients highlights its effectiveness in mitigating privacy-preserving fine-tuning challenges and its scalability to large model architectures in FL environments.

The results of the reasoning task are summarized below, comparing the performance of our proposed model with Vanilla FL, DP-LoRA, and GPT-4o on generated text based on sample questions. For this evaluation, we utilized a fine-tuned model trained on the Alpaca-GPT4 dataset, built upon the \textit{meta-llama/Llama-2-13b-hf} architecture. The findings indicate that our proposed model achieves performance nearly comparable to GPT-4o’s reference outputs on the evaluated questions, while outperforming both Vanilla FL and DP-LoRA. Notably, despite incorporating FHE and pruning for enhanced security and communication efficiency, the proposed method maintains high effectiveness. These results highlight the robustness, privacy-preserving capability, and practical utility of our approach in secure federated fine-tuning of LLMs.

\subsection{Impact of pruning}

\begin{figure*}[t]
  \centering

  \begin{minipage}[b]{0.32\linewidth}
    \centering
    \includegraphics[width=\linewidth]{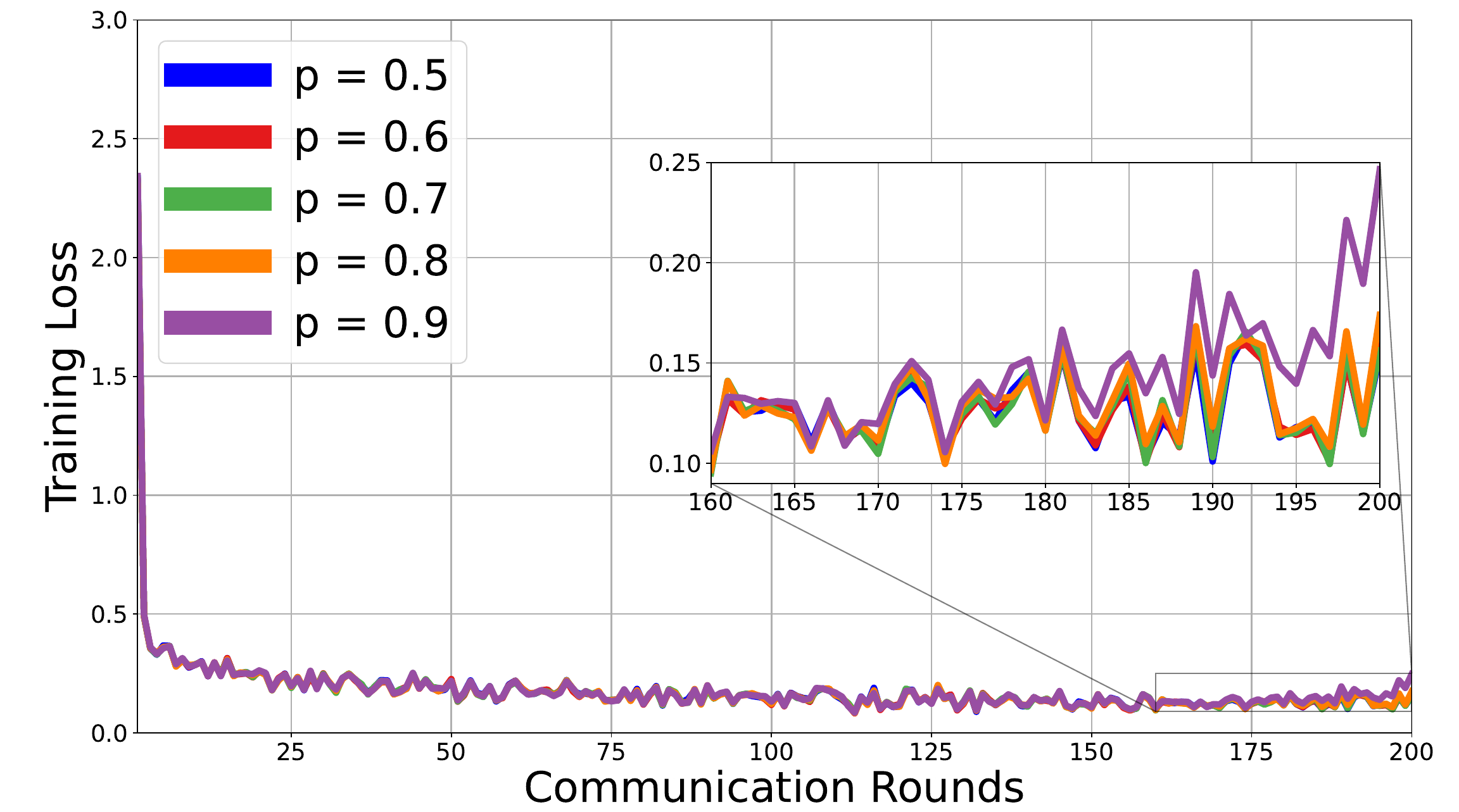}
    \subcaption{$(a)$}
    \label{fig:cl1_fin}
  \end{minipage}
  \hfill
  \begin{minipage}[b]{0.32\linewidth}
    \centering
    \includegraphics[width=\linewidth]{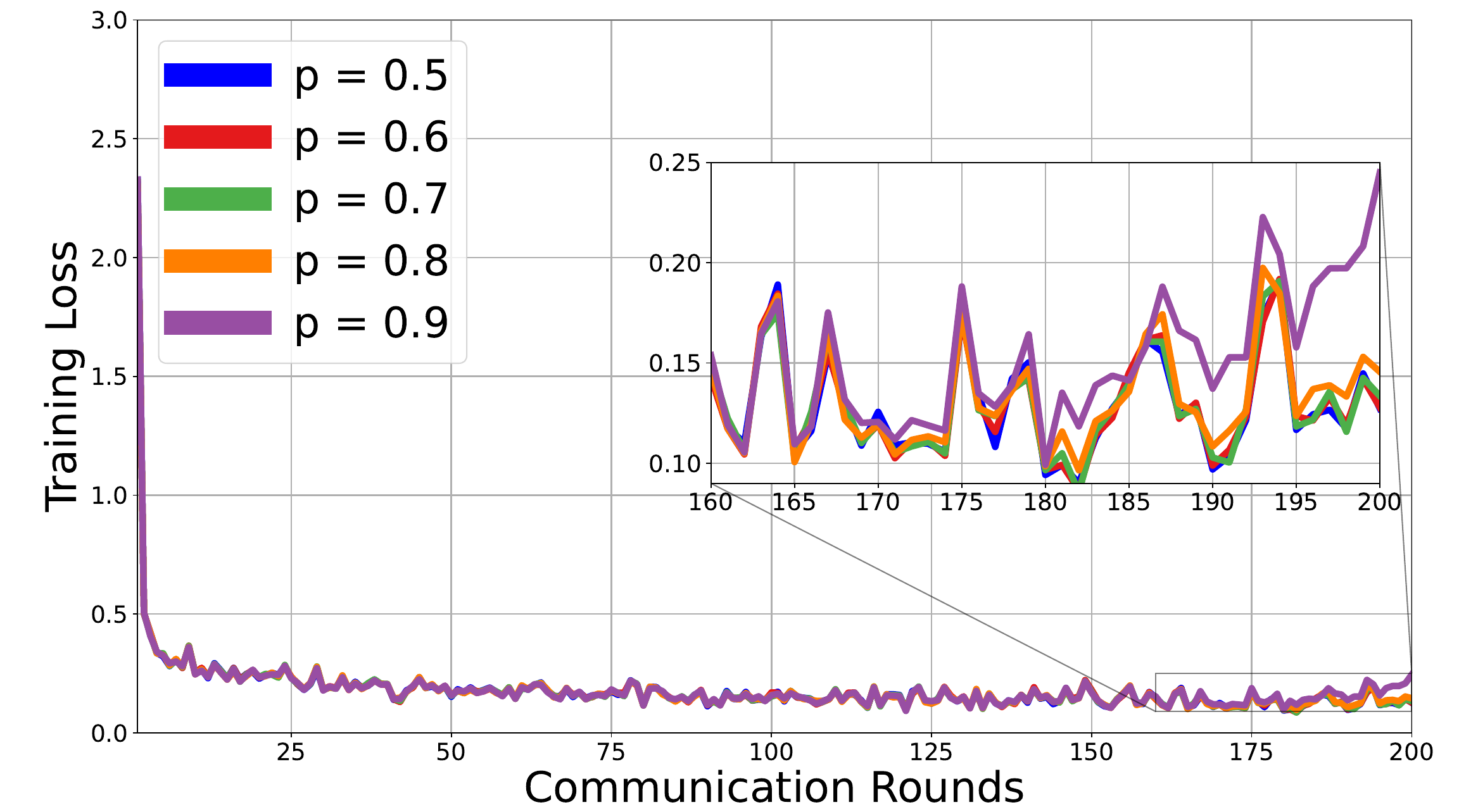}
    \subcaption{$(b)$}
    \label{fig:cl2_fin}
  \end{minipage}
  \hfill
  \begin{minipage}[b]{0.32\linewidth}
    \centering
    \includegraphics[width=\linewidth]{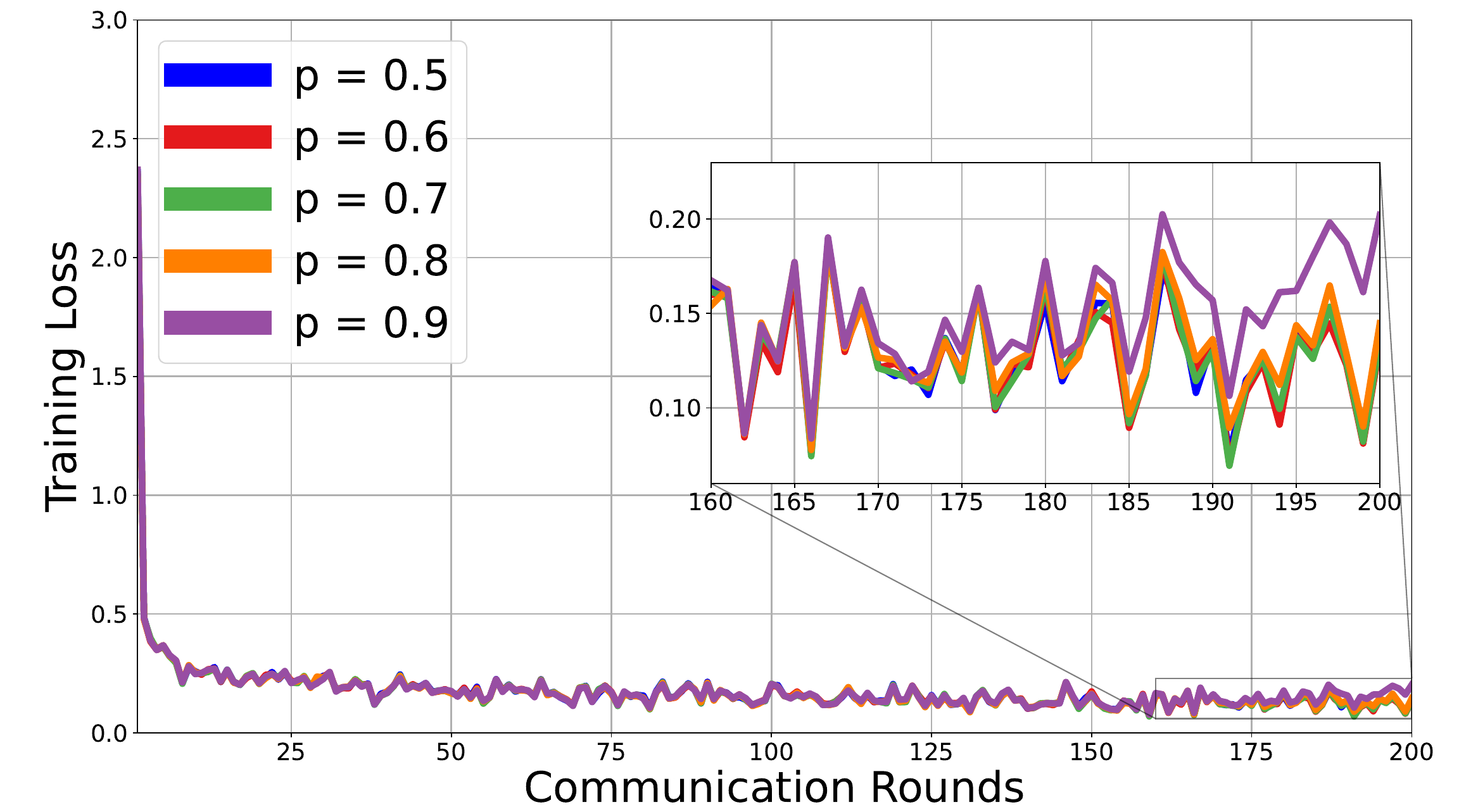}
    \subcaption{$(c)$}
    \label{fig:cl3_fin}
  \end{minipage}

  \caption{
    Training loss of the \textit{meta-llama/Llama-2-13b-hf} model on the 
    \textit{FinGPT/fingpt-sentiment-train} dataset under different pruning rates. 
    Subfigures (a)--(c) correspond to Client 1, Client 2, and Client 3, respectively. 
    Each client is trained on 76,772 financial sentiment samples.}
  
  \label{fig:tr_pruning_rate}
\end{figure*}

We examined the training loss of three clients under pruning rates ranging from 
0.5 to 0.9 (Fig.~\ref{fig:tr_pruning_rate}). The loss remains stable up to a 
pruning rate of 70\%, indicating that model utility is largely preserved. 
Beyond this point, the loss increases noticeably, showing that overly aggressive 
pruning begins to remove essential parameters and hinders convergence.

\begin{table*}[t]
\centering
\caption{ROUGE evaluation results for text reconstruction under different pruning rates using DAGER \cite{petrov2024dager}.}
\label{tab:pruning_rouge_full}
\scriptsize
\renewcommand{\arraystretch}{1.25}

\begin{tabular}{c|ccc|ccc|ccc|ccc|c}
\hline
\textbf{Pruning} & 
\multicolumn{3}{c|}{\textbf{ROUGE-1}} &
\multicolumn{3}{c|}{\textbf{ROUGE-2}} &
\multicolumn{3}{c|}{\textbf{ROUGE-L}} &
\multicolumn{3}{c|}{\textbf{ROUGE-Lsum}} &
\textbf{Average F} \\
\textbf{Rate} &
F & P & R &
F & P & R &
F & P & R &
F & P & R &
$\frac{R1_F + R2_F}{2}$ \\
\hline

No Pruning &
0.5576 & 0.7884 & 0.4925 &
0.5011 & 0.7239 & 0.4470 &
0.5545 & 0.7846 & 0.4897 &
0.5551 & 0.7851 & 0.4904 &
0.5294 \\

0.5 &
0.5878 & 0.8211 & 0.5215 &
0.5380 & 0.7686 & 0.4810 &
0.5858 & 0.8186 & 0.5199 &
0.5861 & 0.8187 & 0.5201 &
0.5629 \\

0.6 &
0.4061 & 0.8857 & 0.2908 &
0.3469 & 0.8393 & 0.2469 &
0.4062 & 0.8849 & 0.2910 &
0.4064 & 0.8852 & 0.2913 &
0.3765 \\

0.7 &
0.4531 & 0.8728 & 0.3441 &
0.3966 & 0.8200 & 0.3013 &
0.4533 & 0.8722 & 0.3444 &
0.4533 & 0.8721 & 0.3445 &
0.4249 \\

\hline
\end{tabular}
\end{table*}

In order to show the impact of pruning, we repeated the original DAGER experiments using different pruning rates. We did not apply any encryption or decryption so that the effect of pruning could be observed in isolation. Table~\ref{tab:pruning_rouge_full} reports the ROUGE-based evaluation of reconstructed text on the Rotten Tomatoes dataset \cite{petrov2024dager}. At a 50\% pruning rate, the F-measure is slightly higher than the non-pruned model. This increase occurs because the reconstructed outputs become shorter and more template-like, which raises precision even though the actual recovery quality is poor. The attacker cannot reproduce full sentences at this level. When the pruning rate increases to 0.6 and 0.7, all ROUGE scores drop, with ROUGE-2 showing the largest decline. This indicates that DAGER struggles to recover meaningful phrases under higher sparsity. Overall, higher pruning rates make gradient inversion attacks significantly less effective.

\begin{center}

\begin{tcolorbox}[colback=gray!10, colframe=black, width=\columnwidth, sharp corners]
\vspace{-0.05in}
\noindent \textbf{Category: Reasoning}

\vspace{0.05in}
\noindent \textbf{Question:} One morning after sunrise, a person was standing facing a pole. The shadow of the pole fell exactly to that person's right. Can you tell me the direction towards which the shadow was pointing — east, south, west, or north? Explain your reasoning steps. To which direction was that person facing? How do you solve this?
\end{tcolorbox}

\vspace{0.5em}

\begin{tcolorbox}[
  colback=pink, colframe=pink, width=\columnwidth,
  boxrule=0pt, arc=2pt, left=6pt, right=6pt, top=6pt, bottom=6pt,
]
\textbf{\faCommentDots\ GPT-4o}\\[0.5em]
In the morning, the sun is in the east, so shadows fall toward the west. Since the shadow of the pole is falling to the person’s right, and the shadow is pointing west, the person must be facing south. Thus, the shadow is pointing west, and the person is facing south.
\end{tcolorbox}

\begin{tcolorbox}[
  colback=yellowgreen, colframe=yellowgreen, width=\columnwidth,
  boxrule=0pt, arc=2pt, left=6pt, right=6pt, top=6pt, bottom=6pt,
]
\textbf{\faCommentDots\ Vanilla-FL}\\[0.5em]
The direction towards which the shadow was pointing is south. The reason for this is that the sun rises in the east and sets in the west, so if the shadow is falling to the right of the person, it means that the sun is in the east, and the person is facing south.

To determine the direction towards which the person was facing, we can use the following steps:

1. Determine the direction of the shadow: The shadow is falling to the right of the person, so the person is facing south.\\
2. Determine the direction of the sun: The sun is in the east, so the person is facing south.\\
3. Confirm the direction: The person is facing south, which is the opposite direction of the sun. This confirms that the person is facing south.
\end{tcolorbox}

\begin{tcolorbox}[
  colback=orange, colframe=orange, width=\columnwidth,
  boxrule=0pt, arc=2pt, left=6pt, right=6pt, top=6pt, bottom=6pt,
]
\textbf{\faCommentDots\ DP-LoRA}\\[0.5em]
The direction towards which the shadow was pointing is south.

To solve this, we need to understand the direction of the sun's movement. The sun rises in the east and sets in the west. During the day, the sun moves from east to west, and its shadow moves in the opposite direction.

Since the shadow of the pole fell to the person's right, it means that the sun was in the east, and the person was facing east.

The person was facing east.
\end{tcolorbox}

\begin{tcolorbox}[
  colback=yellowgreen, colframe=yellowgreen, width=\columnwidth,
  boxrule=0pt, arc=2pt, left=6pt, right=6pt, top=6pt, bottom=6pt,
]
\textbf{\faCommentDots\ FedShield-LLM}\\[0.5em]
The direction towards which the shadow was pointing is south. The reason for this is that the sun rises in the east and sets in the west, so if the shadow is falling to the right of the person, it means that the sun is in the east, and the person is facing south.

To determine the direction towards which the person was facing, we can use the following steps:

1. Identify the direction of the shadow: The shadow is falling to the right of the person, so the person is facing south.\\
2. Identify the direction of the sun: The sun is in the east, so the person is facing south.\\
3. Confirm the direction: The person is facing south, as the sun is in the east and the shadow is falling to the right.
\end{tcolorbox}

\end{center}

Due to space limitations, we provide detailed results for other question categories, including reasoning, generic, knowledge, common-sense, Fermi, and coding tasks, in the supplementary materials. These comparisons emphasize the versatility and consistency of our proposed model across a wide range of domains and question types, demonstrating its superior performance in generating accurate and coherent responses.

\begin{table}[h!]
\centering
\caption{Average BERTScore comparison across different methods while GPT-4o is baseline.}
\begin{tabular}{lccc}
\toprule
\textbf{Model} & \textbf{Precision} & \textbf{Recall} & \textbf{F1 Score} \\
\midrule
Vanilla-FL     & 0.5683 & 0.5867 & 0.5756 \\
DP-LoRA        & 0.6287 & 0.6031 & 0.6130 \\
FedShield-LLM  & \textbf{0.6738} & \textbf{0.7012} & \textbf{0.6865} \\
\bottomrule
\end{tabular}
\label{tab:bertscore_comparison}
\end{table}

For this evaluation, we used seven questions presented in the main paper and supplementary material. The response quality of Vanilla-FL, DP-LoRA, and FedShield-LLM was assessed using BERTScore, which computes semantic similarity between generated responses (candidates) and GPT-4o outputs (references). Specifically, we employed the pre-trained \texttt{bert-base-uncased} model for English text to obtain precision, recall, and F1 scores for each response. As shown in Table~\ref{tab:bertscore_comparison}, FedShield-LLM achieved the highest average F1 score of 0.6865, outperforming both DP-LoRA (0.6130) and Vanilla-FL (0.5756). In terms of precision and recall, FedShield-LLM also led with scores of 0.6738 and 0.7012, respectively, compared to DP-LoRA's 0.6287 precision and 0.6031 recall, and Vanilla-FL’s 0.5683 precision and 0.5867 recall. While GPT-4o a multimodal and significantly larger model achieves a perfect BERTScore self-F1 of 1.0 on its own outputs, these results indicate that FedShield-LLM, a lightweight and privacy-preserving 7B and 13B parameter model, can produce responses that are nearly comparable to GPT-4o’s within the evaluated scope. This underscores the effectiveness of our approach in semantically aligning with high-quality responses on specialized tasks, despite substantial differences in model scale and architecture.

Our proposed model defends against inference attacks as the server can only access encrypted LoRA parameters. In this case, the server has no knowledge of the actual model parameters. However, even if we allow the server to decrypt the model, an honest-but-curious server with access to the model parameters will still be unable to infer sensitive information through gradient inversion or reverse engineering attacks. This is because the attacker will only have access to the sparsified model. Therefore, our proposed model provides robust security against adversaries.

%% file: Discussion.tex
\section{Discussion} \label{sec:Discussion}

The results of this study underscore the effectiveness of FedShield-LLM in enhancing the performance and security of LLMs in FL. With the same hyperparameters and dataset, FedShield-LLM consistently outperformed Vanilla federated LLM and DP-LoRA, achieving lower training loss and generating text of higher quality, nearly comparable to GPT-4o. Our method demonstrated superior text generation across diverse question types and proved to be particularly suitable for cross-silo environments with resource constraints. By leveraging parameter-efficient fine-tuning through LoRA, FedShield-LLM significantly reduces computational and memory requirements, making it practical for environments with limited resources.

Additionally, the integration of FHE with unstructured pruning optimized model parameters while ensuring robust data privacy, addressing critical challenges in secure distributed LLM training. In particular, incorporating FHE enables secure aggregation over encrypted model updates, preventing the server from accessing plaintext LoRA parameters and mitigating risks such as gradient inversion and membership inference attacks. This is especially important for LLMs, where model updates can implicitly encode sensitive information. However, applying FHE directly to full LLM fine-tuning is impractical due to the extremely large number of parameters and the high computational and communication overhead associated with encrypted operations. To make FHE practical in this setting, our framework leverages LoRA-based parameter-efficient fine-tuning together with unstructured pruning to significantly reduce the size and information content of model updates prior to encryption.

To our knowledge, this is the first implementation of such an approach, positioning FedShield-LLM as a robust and efficient framework for sensitive and resource-constrained FL applications. Nevertheless, this approach still introduces some computational complexity due to FHE operations and may be influenced by data heterogeneity across clients. Despite the proposed optimizations, FHE still introduces non-negligible computational cost, memory overhead, and ciphertext expansion, which may affect scalability as model size and the number of participating clients increase. In future work, we will explore more efficient designs to enable the practical application of FHE for full LLM fine-tuning. Future work will focus on reducing this overhead and designing strategies that better accommodate highly non-IID client data distributions. We will also explore more scalable encrypted aggregation designs and lightweight optimization methods to broaden the applicability of FedShield-LLM in real-world deployments.